\documentclass[11pt,letterpaper,twoside]{article}

\usepackage[letterpaper,margin=0.75in,includeheadfoot,headheight=14pt]{geometry}
\usepackage{fancyhdr}

\usepackage{amsmath,amssymb,amsthm}

\usepackage{graphicx}
\usepackage[dvipsnames]{xcolor}
\usepackage{float}
\usepackage{booktabs}
\usepackage{multirow}
\usepackage{makecell}
\usepackage{caption}
\usepackage{enumitem}
\usepackage{pdflscape}
\usepackage{setspace}
\usepackage{url}
\usepackage{parskip}
\usepackage{ragged2e}
\usepackage[
  colorlinks=true,
  linkcolor=Blue,
  citecolor=Blue,
  urlcolor=Blue
]{hyperref}

\usepackage[authoryear,round]{natbib}
\usepackage{bibunits}

\newcommand{\ManuscriptTitle}{Additive Nonparametric Regression with Spatial and Network Objects}
\newcommand{\RunningTitle}{Additive Nonparametric Regression with Spatial and Network Objects}
\newcommand{\RunningAuthors}{Guhaniyogi, Dey, Chandra, Scheffler, and Mallick}

\title{\ManuscriptTitle}
\author{%
  Rajarshi Guhaniyogi\textsuperscript{1},
  Pritam Dey\textsuperscript{1},
  Krishnendu Chandra\textsuperscript{2},\\
  Aaron Scheffler\textsuperscript{3}, and
  Bani K. Mallick\textsuperscript{1}\\[0.5em]
  \small \textsuperscript{1}Department of Statistics, Texas A\&M University, College Station, TX, USA\\
  \small \textsuperscript{2}School of Medicine, Indiana University, Indianapolis, IN, USA\\
  \small \textsuperscript{3}Division of Biostatistics and Epidemiology,\\
  \small University of California, San Francisco, CA, USA
}
\date{}

\newtheorem{theorem}{Theorem}
\newtheorem{lemma}{Lemma}
\newtheorem{result}{Result}

\newcommand{\ba}{\boldsymbol{a}}
\newcommand{\bA}{\boldsymbol{A}}
\newcommand{\bg}{\boldsymbol{g}}
\newcommand{\bG}{\boldsymbol{G}}
\newcommand{\bI}{\boldsymbol{I}}
\newcommand{\bL}{\boldsymbol{L}}
\newcommand{\bo}{\boldsymbol{o}}
\newcommand{\bO}{\boldsymbol{O}}

\newcommand{\bu}{\boldsymbol{u}}
\newcommand{\bU}{\boldsymbol{U}}
\newcommand{\bv}{\boldsymbol{v}}
\newcommand{\bw}{\boldsymbol{w}}
\newcommand{\bW}{\boldsymbol{W}}
\newcommand{\bx}{\boldsymbol{x}}
\newcommand{\bbeta}{\boldsymbol{\beta}}
\newcommand{\bepsilon}{\boldsymbol{\epsilon}}
\newcommand{\bPhi}{\boldsymbol{\Phi}}
\newcommand{\bmu}{\boldsymbol{\mu}}
\newcommand{\bet}{\boldsymbol{\eta}}
\newcommand{\bzero}{\boldsymbol{0}}

\AtBeginDocument{%
  \setlength{\abovedisplayskip}{6pt plus 2pt minus 2pt}
  \setlength{\belowdisplayskip}{6pt plus 2pt minus 2pt}
  \setlength{\abovedisplayshortskip}{3pt plus 2pt minus 1pt}
  \setlength{\belowdisplayshortskip}{3pt plus 2pt minus 1pt}
}

\begin{document}

\begin{bibunit}[plainnat]

\maketitle
\thispagestyle{plain}

\begin{abstract}
\noindent
This article is motivated by an imaging application from the Adolescent Brain Cognitive Development (ABCD) study, aiming to predict task-based brain activation maps from t-fMRI using cortical metrics from structural MRI (s-MRI) and brain connectivity data from resting-state fMRI (rs-fMRI). Hierarchical Bayesian modeling is well-suited for integrating diverse imaging data and quantifying prediction uncertainty. However, progress in this field is limited due to challenges in designing joint priors that capture the structures and interconnections between different imaging modalities, along with computational complexity and lack of theoretical assurances. To address these challenges, the article introduces a novel regression framework that treats t-fMRI and s-MRI images as functional data, incorporating additive non-linear effects of both network and functional predictors on the functional response. Specifically, we employ Gaussian process (GP) priors on coefficients related to the functional predictors to capture their intricate functional dependencies with the response. Furthermore, a GP prior is assigned to encapsulate the non-linear nodal effects of the network predictor on the response function. The method is supported by theoretical results on predictive accuracy for the functional response, and is empirically validated through simulation studies and analysis of multi-modal neuroimaging data from the ABCD study. Additional details regarding model computation, posterior consistency of the proposed model and empirical results are available in the supplementary material.
\end{abstract}


\setstretch{1.6}

\section{Introduction}
This article considers localization of tasked-evoked brain activation that can allow researchers to map regions of interest (ROI) associated with cognitive domains, identify biomarkers associated with neurocognitive disorders, and map functional regions for neurosurgical planning \citep{jones2017}. Despite the research and clinical utility of task-based activation maps, several works have found that task-based fMRI (t-fMRI) activation maps have limited reliability and accuracy and thus it is of interest to determine how well individual task mapping can be recovered from more stable features such as morphological features (e.g. cortical thickness and sulcal depth) captured via structural MRI (s-MRI) and resting state connectivity networks captured via fMRI (rs-fMRI) \citep{weng2018, elliot2020}. 
Resting state connectivity networks has edges describing connectivity among \emph{functional nodes} (brain subnetworks, e.g. auditory) \citep{gordon2017}, and operate at a different scale than t-fMRI and s-MRI, which are defined within spatially indexed ROIs. However, these modalities are topologically linked via the brain’s hierarchical organization (see Figure~\ref{fig1}), with ROIs nested within functional nodes (subnetworks). Statistically, this setup is formulated as a regression problem where the spatial t-fMRI activation map is the response, regressed against spatially indexed s-MRI and the brain network. The primary goal is to model this regression to capture non-linear associations and enable robust predictive inference for image responses.

\begin{figure}
\centering
\includegraphics[width=\textwidth]{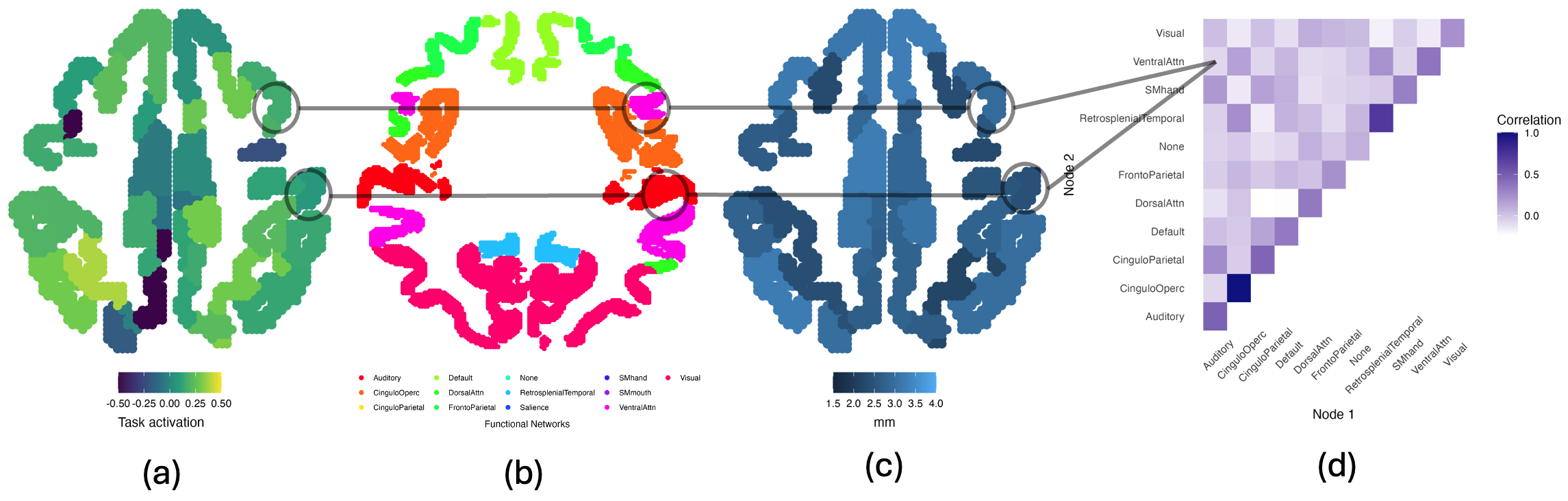}
\caption{\footnotesize{{Schematic of the multi-object brain imaging data structure for a sample subject. (a) Structural image encoding ROI-level task activation, (b) Gordon functional atlas parcellation of the brain into functional subnetworks, (c) structural image encoding ROI-level cortical thickness (mm) (d) Network image obtained by calculating the pairwise Pearson correlation for the average within and between subnetworks. Grey circles and lines connect (a, c) structural and (d) network information from images via the (b) parcellated atlas. Thus, the atlas provides an organizing hierarchy that links together structural information (task activation) at the ROI-level with network information indexed by pairs of functional networks (fMRI).}}}
\label{fig1}
\end{figure}

The literature on regression models involving images can be divided into three main areas: (a) scalar-on-image regression, (b) image-on-scalar regression, and (c) image-on-image regression. In scalar-on-image regression, images serve as predictors for a scalar outcome, while in image-on-scalar regression, images are the response variable. Both approaches have been extensively studied, with images treated as spatially correlated functions \citep{goldsmith2014smooth, feng2019bayesian, lin2024latent, roy2021spatial, kang2018scalar}, or as tensors and networks \citep{guhaniyogi2017bayesian, guha2020bayesian, guha2024covariate}. A less explored but highly relevant area is image-on-image regression, which focuses on studying associations between imaging modalities, closely aligned with the goals of our application in Section~\ref{data application}.

In image-on-image regression, adaptive smoothing and spatially varying coefficient models have been introduced, allowing for better integration of neighboring voxel information \citep{niyogi2023tensor, mu2018estimation}. Spatial latent factor models have also been developed to capture complex spatial dependencies \citep{guo2022spatial}. Although these techniques are effective, they do not account for situations where spatial predictors are defined at ROIs while the network predictor is constructed using nodes that represent brain sub-networks. Since each ROI is nested within a sub-network, this leads to a misalignment between the spatial and network predictors. A second line of research explores multivariate support vector machines for predicting missing spatial or temporal information in fMRI and EEG data \citep{de2011predicting, jansen2012motion}. Recently, deep learning methods, particularly adversarial networks, have been applied to learn mappings between image inputs and output for translation, style transfer, and data augmentation \citep{huang2018, ma2023}.

This article introduces Additive Multi-Object Gaussian Process (AMO-GP), a two-stage semi-parametric regression framework to address image-on-image regression, where predictor images include spatially-varying functions and a network. In the first stage, each network edge is modeled as a bi-linear effect of latent vectors corresponding to the two connected nodes to estimate node-specific latent effects. In the second stage, we employ a semi-parametric regression framework, assuming additive effects of various image predictors on the response. The effect of spatially-indexed functional predictors are modeled with spatially varying coefficient functions, while effect of each node from the network predictor are represented through a non-linear function of the estimated latent effects from the first stage. Gaussian process (GP) priors are assigned to the spatially varying coefficients and the function capturing node-specific latent effects to enable information sharing across ROIs and subjects. An additive framework provides interpretable insights into the influence of each predictor, facilitates regularization, as modeling only additive effects reduces model complexity, particularly when data are limited. Also, the additive structure circumvents the need to specify joint covariance kernels across all predictors. This enables flexible integration of diverse image data while ensuring accurate uncertainty quantification for inference.

The rest of the manuscript is organized as follows. Section~\ref{data application} overviews the multi-modal neuroimaging data and outlines the main scientific questions. Section~\ref{model_development} details model development and the prior, while Section~\ref{posterior_comp} discusses posterior computation. Section~\ref{sim_studies} presents simulation studies validating AMO-GP, and Section~\ref{multi_modal_data_analysis} showcases its application to multimodal neuroimaging data. Section~\ref{conclusion} concludes the paper. 

\section{Multi-modal Neuroimaging Data Application}\label{data application}

\noindent We explore a clinical application derived from multimodal imaging studies in 9–10 year-old children from the Adolescent Brain Cognitive Development (ABCD) Study, the largest U.S. study on brain development \citep{casey2018adolescent}. The focus is on characterizing the association between task-based activation maps measuring local neuronal activity during tasks, and brain image predictors capturing cortical morphology and resting state connectivity.

\noindent\underline{\textbf{Scientific question of interest:}} Task-based activation maps are based on fMRI images of local blood-oxygen-level-dependent (BOLD) signals collected on subjects engaged in a task \citep{ellis2022}. Localization of tasked-evoked responses allows researchers to map regions of interest (ROIs) associated with cognitive domains, identify biomarkers associated with neurocognitive disorders, and map functional regions for neurosurgical planning \citep{jones2017}. Despite the clinical utility of task-based activation maps, several works have found that individual task-based fMRI (t-fMRI) activation maps may display low reliability and thus it is of interest to determine how well task mapping can be recovered from more stable features such as morphological features and resting state connectivity maps \citep{weng2018, elliot2020}. Specifically, even if individual t-fMRI activation maps display low reliability, it may be possible to recover expected associations between t-fMRI signals and more stable neurological markers by analyzing images from multiple subjects. For example, \cite{cole2016} demonstrated that coordinated task-based activation follows resting state functional connectivity networks which suggests resting state functional networks could be used to help triangulate task activation in the future. Our focus is to model task-based activation in response to a working memory task as a function of structural and network images. Working memory, which underlies the temporary storage and manipulation of information, is critical for cognition and has been extensively studied in relation to cognitive development and neural substrates \citep{rosenberg2020}. Specifically, we analyze task-based activation from the working memory task (described below) as a function of cortical morphological features and resting state connectivity, as both have been previously linked to working memory in childhood \citep{osaka2021}. 

Most existing models relate t-fMRI activation maps to brain structure or function using features from a single imaging modality. Associations with sMRI have been modeled using regression \citep{squeglia2013}, deep learning \citep{ellis2022}, and latent factor models \citep{guo2022spatial}, while rs-fMRI connectivity links have leveraged linear models \citep{tavor2016, cohen2020}, permutation tests \citep{harrewijn2020}, neural networks \citep{cohen2020, ngo2022}, ensemble methods \citep{cohen2020, zheng2022}, and manifold learning \citep{langs2015}. Joint analysis of t-fMRI activation with both s-MRI and rs-fMRI features is rare: vectorized feature approaches \citep{tavor2016} do not clarify structural or network contributions, and spatial-network models \citep{ma2023} obscure additive effects. Here, we integrate cortical morphology and functional connectivity to predict task activation, and, to our knowledge, no previous structured regression has jointly modeled t-fMRI, s-MRI, and rs-fMRI associations.

\noindent\underline{\textbf{Clinical images and working memory evaluation:}} We utilized the ABCD Study 5.0 Tabulated Release Data (\url{http://dx.doi.org/10.15154/8873-zj65}), focusing on baseline imaging and cognitive measures from a random subsample of 80 children (aged 9–10 years) from a single study site. Sample size is restricted due to computational constraints (see Section~\ref{posterior_comp}). The imaging data included: task-based fMRI (t-fMRI) for measuring focal brain activation, sMRI for cortical morphological features, and resting state fMRI (rs-fMRI) for assessing brain activation through neuronal oxygen consumption during rest  \citep{hagler2019}.

All images are registered in the template space of the Montreal Neurological Institute (MNI). Working memory was measured using fMRI activation (Figure \ref{fig1}a) of the emotional N-back task using the 2-back vs. 0-back linear contrast described in \citep{hagler2019} for 148 spatially indexed ROIs defined by the Destrieux atlas \citep{destriuex2010}. Morphological features include cortical thickness and sulcal depth which estimate the thickness of the cortex and depths of folds in the cortex, respectively, (Figure \ref{fig1}) and were measured using sMRI and extracted using Freesurfer for the Destrieux atlas (described above) as in \citep{hagler2019}. Resting state fMRI data was collected for 333 cortical-surface ROIs defined by the Gordon atlas \citep{gordon2017} and ROIs were subsequently grouped into 13 brain subnetworks (auditory, cingulo-opercular, cingulo parietal, default mode, dorsal attention,  frontoparietal, none, retrosplenial temporal,  salience, sensorimotor-hand, sensorimotor-mouth, ventral attention and visual subnetworks \citep{gordon2017}). A symmetric adjacency matrix was constructed following \citep{hagler2019}, with rows and columns representing different subnetworks. Matrix entries correspond to Z-scores obtained by Fisher Z-transforming the average Pearson correlation for each pairwise combination of ROIs between subnetworks (Figure \ref{fig1}d).

To align the Destrieux and Gordon atlases with differing parcellations, each Destrieux ROI was assigned to a Gordon subnetwork by minimizing the Euclidean distance between their centroids. This enabled cross-modality and subject comparisons. Each Destrieux ROI was thus nested within a Gordon subnetwork (Figure \ref{fig1}b). Centroid matching was used given the ROI-level focus of the analysis, with coordinates obtained via the R package brainGraph \citep{brainGraph}. The median distance between matched ROIs was 6.6 mm (range: 1.4–17.2 mm), and visual checks confirmed accuracy. This alignment nested 145 Destrieux ROIs within 11 Gordon subnetworks, excluding salience and sensorimotor-mouth networks due to sparse coverage. Atlas selection and alignment sensitivity are discussed in Section~\ref{conclusion}.

\section{Model and Prior Development}\label{model_development}

\subsection{Multi-object image characterization}

We propose Additive Multi-Object Gaussian Process (AMO-GP) a multi-object framework which jointly captures both spatial and network information. AMO-GP represents each image either as a spatially indexed function or as a network object, linking them through a shared topology. For the $i$\textsuperscript{th} individual, let $\bA_i$ denote the undirected brain network, represented by a $P\times P$ symmetric matrix derived from rs-fMRI data. The $P$ nodes correspond to brain subnetworks $\mathcal{R}_1,\ldots,\mathcal{R}_P$ with the $(p,p')$-th entry $a_{i,(p,p')}$ indicating the strength of association between $\mathcal{R}_p$ and $\mathcal{R}_{p'}$. The network is symmetric (i.e. $a_{i,(p,p')}=a_{i,(p',p)}$, for all $1 \leq p,p' \leq P$) and excludes self-connections ($a_{i,(p,p)}=0$). Assume $V_p$ ROIs are nested within subnetwork $\mathcal{R}_p$. Functional predictors $\bg_{i,p}$, $\bw_{i,p}$ along with the functional response $\bo_{i,p}$ are defined over these ROIs for $p=1,\ldots,P$. Specifically, the values for the $\bv=(v_1, \ldots, v_D)$\textsuperscript{th} ROI coordinate are given by $g_{i,p}(\bv)$, $w_{i,p}(\bv)$, and $o_{i,p}(\bv)$, respectively, where $D$ depends on whether the analysis focuses on a cross-section ($D=2$) or the entire brain ($D=3)$. 

\subsection{Model Framework}\label{mf}

AMO-GP is a two-stage Bayesian approach that blends ideas from manifold regression for networks and non-parametric modeling of regression effects from the functional predictors. A two-stage framework is adopted for computational feasibility and to avoid slow convergence of latent variables. Since the latent scale representation effectively captures the underlying network topology, the two-stage approach in AMO-GP does not result in any noticeable loss in predictive uncertainty when applied to the ABCD data for predicting t-fMRI images.
\subsubsection{Stage 1: Latent scale representation of brain networks} \label{mf1}
As it's first stage, AMO-GP models the undirected network predictor $\mathbf{A}_i$ through its edge set $\mathbf{a}_i = \{a_{i,(p,p')}: 1 \leq p < p' \leq P\}$ for $i = 1, \ldots, n$. Since the number of edges $P(P-1)/2$ grows rapidly with $P$, we adopt a low-dimensional latent representation to capture network structure. Specifically, we consider
\begin{align}\label{latent_scale}
\eta^{-1}(E[a_{i,(p,p')}|\mu_i,\tilde{\bu}_{i,p},\tilde{\bu}_{i,p'},\sigma_u^2])=\mu_i+\tilde{\bu}_{i,p}^{T}\tilde{\bu}_{i,p'}+\epsilon_{i,(p,p')},\qquad \epsilon_{i,(p,p')}\sim N(0,\sigma_u^2),
\end{align}
where $\eta(\cdot)$ is a suitable link function. Here, $\mu_i$ is a subject-specific intercept controlling overall network density, and $\tilde{\bu}_{i,p}\in\mathbb{R}^R$ is a latent vector for node $\mathcal{R}_p$. The inner product $\tilde{\bu}_{i,p}^T \tilde{\bu}_{i,p'}$ captures connectivity, reducing the parameter dimension from $P(P-1)/2$ to $PR$. Although methods for selecting \citep{hoff2005bilinear, guhaniyogi2016compressed} or inferring \citep{guha2020bayesian} $R$ exist, they are computationally intensive. Given the stability of predictive performance across a wide range of $R$, we fix $R$ in AMO-GP.

The model in \eqref{latent_scale} is closely related to the random dot product graph model \citep{xie2023efficient}. In latent space modeling of a single network, this representation provides a sufficiently general characterization of interconnection structure. We assign a flat prior on each $\mu_i$ and the priors $\tilde{\bu}_{i,p} \stackrel{i.i.d.}{\sim} N(\bzero,\bI_R)$ and $\sigma_u^2 \sim IG(\alpha,\beta)$, and perform posterior inference via  Markov Chain Monte Carlo (MCMC). Since the likelihood in \eqref{latent_scale} depends on $\tilde{\bu}_{i,p}$'s only through inner products, these latent variables are identifiable up to orthogonal transformations. We therefore apply a ``Procrustean transformation" \citep{guha2020bayesian} to align MCMC samples to a common reference orientation across subjects and treat the mean of the transformed post-burn-in draws as a point estimate $\widehat{\bu}_{i,p}$ of ${\bu}_{i,p}$. These point estimates are subsequently used as covariates in the second-stage regression.

\subsubsection{Stage 2: Additive Semi-parametric Regression Framework} \label{mf2}
In the second stage of AMO-GP, we formulate a non-linear regression function to delineate association between the functional response $o_{i,p}(\bv)$ and functional predictors $g_{i,p}(\bv)$, $w_{i,p}(\bv)$, and the estimated latent node effect $\widehat{\bu}_{i,p}$ corresponding to the $p$th node/subnetwork. The unknown regression function $f_{p,\bv}$, characterizing the association between response and these predictors in the $\bv$th ROI of the $p$th subnetwork, is decomposed into three components, $f_{p,\bv}=f_p^{N}+f_{p,\bv}^{G}+f_{p,\bv}^{W}$. Here, $f_p^{N}$, $f_{p,\bv}^{G}$, and $f_{p,\bv}^{W}$ represent the non-linear functional effects of $\widehat{\bu}_{i,p}$, $g_{i,p}(\bv)$, and $w_{i,p}(\bv)$, respectively. More precisely, we model:
\begin{align}\label{additive_reg}
o_{i,p}(\bv)=f_p^{N}(\widehat{\bu}_{i,p})+f_{p,\bv}^{G}(g_{i,p}(\bv))+f_{p,\bv}^{W}(w_{i,p}(\bv))+\epsilon_{i,p}(\bv),
\end{align}
where the idiosyncratic errors are assumed to follow $\epsilon_{i,p}(\bv)\stackrel{i.i.d.}{\sim} N(0,\tau_p^2)$.
In AMO-GP, we assume that $f_{p,\bv}^{G}(g_{i,p}(\bv))$ and $f_{p,\bv}^{W}(w_{i,p}(\bv))$ are linear but have coefficients that vary for each ROI, i.e.,  $f_{p,\bv}^{G}(g_{i,p}(\bv))=g_{i,p}(\bv)\bbeta_{p,g}(\bv)$ and $f_{p,\bv}^{W}(w_{i,p}(\bv))=w_{i,p}(\bv)\bbeta_{p,w}(\bv)$. This formulation is an instance of a varying-coefficient model (VCM), a flexible and popular extension of the linear regression model \citep{guhaniyogi2023distributed, guhaniyogi2020large}. Under VCM specification for functional predictors equation \eqref{additive_reg} simplifies:
\begin{align}\label{additive_VCM}
o_{i,p}(\bv)=f_p^{N}(\widehat{\bu}_{i,p}) + g_{i,p}(\bv)\beta_{p,g}(\bv)+ w_{i,p}(\bv)\beta_{p,w}(\bv)+\epsilon_{i,p}(\bv).
\end{align}
We assign Gaussian process priors independently to estimate $\beta_{p,g}(\cdot)$, $\beta_{p,w}(\cdot)$, and $f_p^{N}(\cdot)$, each with a mean of $0$ and exponential covariance kernels $\kappa_{p,g}$, $\kappa_{p,w}$, and $\kappa_{p,N}$ respectively, defined as: 
\begin{equation}
\label{eq:covariance_kernel}
\begin{gathered}
\kappa_{p,g}(\bv,\bv')=\sigma_{p,g}^2\exp(-\theta_{p,g}||\bv-\bv'||), \qquad
\kappa_{p,w}(\bv,\bv')=\sigma_{p,w}^2\exp(-\theta_{p,w}||\bv-\bv'||),\\
\kappa_{p,N}(\widehat{\bu}_{i,p},\widehat{\bu}_{i',p})=\sigma_{p,N}^2\exp(-\theta_{p,N}||\widehat{\bu}_{i,p}-\widehat{\bu}_{i',p}||),
\end{gathered}   
\end{equation}
where $||\cdot||$ denotes the Euclidean distance between two vectors, and $\sigma_{p,g}^2$, $\sigma_{p,w}^2$, and $\sigma_{p,N}^2$ are variance parameters, while $\theta_{p,g}$, $\theta_{p,w}$, and $\theta_{p,N}$ are length-scale parameters. The length-scale parameters control the smoothness of the regression functions. Note, due to the alignment of latent effects $\widehat{\bu}_{i,p}$ to a common orientation as described in Stage 1, distance based kernels are appropriate for modeling the network effects. 

The use of GP priors for the mean function allows for flexible modeling of non-linear relationships between the outcome and the covariates. Furthermore, it enables information sharing between ROIs within a subnetwork and between subjects, facilitating accurate inference with uncertainties for the regression function and for prediction of the response image. The exponential kernel for GP is a specific case of the Matérn class of correlation functions with smoothness parameter $1/2$, yielding spatial surfaces that are continuous but not differentiable a priori \citep{stein1999interpolation}. While it is technically possible to use other Matérn kernels with higher smoothness parameters to obtain smoother functions, the exponential kernel is chosen due to its popularity in spatial analysis. Importantly, methods that perform well with the exponential kernel generally also perform well with Matérn kernels featuring higher smoothness parameters \citep{stein1999interpolation}. The second stage of AMO-GP can be parallelized across different subnetworks, yielding significant computational benefits. Theoretical properties, in particular the posterior consistency of AMO-GP have been studied in Supplementary Section~\ref{posterior_consis}.

\section{Posterior Computation: Parallel Computation Over Subnetworks}\label{posterior_comp}

We will address the posterior computation for the two stages of AMO-GP separately.

\noindent\underline{\textbf{Stage 1.}} In the empirical studies, we use the identity link for $\eta(\cdot)$ since the edge weights in the brain network are continuous in the ABCD data. The intrinsic dimensionality is fixed at $R=6$, consistent with prior studies showing minimal sensitivity in inference when $R$ is chosen around $5$ \citep{gutierrez2023bayesian}. A Gibbs Sampler is implemented (see Section~\ref{posterior-deriv} of the Supplementary materials) for 15,000 iterations, with the first 10,000 discarded as burn-in. Estimates for $\widehat{\bu}_{i,p}$ are derived from the remaining 5,000 iterations.

\noindent\underline{\textbf{Stage 2.}} The posterior computation in the second stage of AMO-GP involves elliptical slice sampling for length-scale parameters $\theta_{p,g}$, $\theta_{p,w}$, and $\theta_{p,N}$ and Metropolis–Hastings algorithm for the variance parameters $\tau_p^2$, $\sigma_{p,g}^2$, $\sigma_{p,w}^2$ and $\sigma_{p,N}^2$ to generate $15000$ MCMC samples, discarding the first $10000$ as burn-in. Details of the full conditional distributions are available in Section~\ref{posterior-deriv} of the Supplementary materials.

\noindent\underline{\textbf{Parallel computation and computation complexity.}} The computational complexity of AMO-GP is primarily driven by two factors: updating the varying coefficients $(\beta_{p,g}(\bv):\bv\in\mathcal{R}_p)^T$ and $(\beta_{p,w}(\bv):\bv\in\mathcal{R}_p)^T$ across all ROIs for each subnetwork, and updating
$(f_p^N(\widehat{\bu}_{1,p}),\ldots,f_p^N(\widehat{\bu}_{n,p}))^T$ for each subnetwork during each MCMC iteration. Since the method allows for parallel processing across different subnetworks, the complexity of updating the varying coefficients and node effects is approximately $\sim\max_{p=1,..,P} V_p^3$ and $\sim n^3$, respectively. 
To ensure feasible computation, we use a moderate sample size, with Section~\ref{conclusion} discussing future work to address challenges with larger datasets. All implementations are performed in the \texttt{R} programming language.

\section{Simulation Studies}\label{sim_studies}

\subsection{Simulated Data Generation}\label{simdatgen}

We consider $P=50$ subnetworks and a sample size of $n=100$ for both training and test datasets. Each subnetwork contains 12 regions of interest (ROIs), i.e., $V_p=12$, for $p = 1, \ldots, P$. The network predictor is constructed from $R=3$ latent components. For each subject $i=1,\ldots,n$ and subnetwork $p=1,\ldots,P$, we generate a network-level predictor based on latent variables $\tilde{\bu}_{i,p}$, along with two functional predictors $w_{i,p}(\bv)$ and $g_{i,p}(\bv)$ defined over spatial locations $\bv \in [0,10]^3$. The response image $o_{i,p}(\bv)$ is generated under four scenarios summarized in Table~\ref{tab:sim_scenarios}, designed to assess performance under both well-specified (Scenarios 1 and 3) and misspecified (Scenarios 2 and 4) settings. All experiments are repeated over 10 independent replicates. Details are provided in Supplementary Section~\ref{supp:simdatgen}.
\begin{table}[!t]
\centering
\small
\caption{\footnotesize Summary of simulation scenarios.}
\label{tab:sim_scenarios}
\begin{tabular}{c p{0.23\linewidth} p{0.6\linewidth}}
\toprule
\textbf{Scenario} & \textbf{Type} & \textbf{Data Generating Model Specification} \\
\midrule
1 & AMO-GP additive model 
& Additive AMO-GP model as in \eqref{additive_VCM} with Gaussian noise; $f_p^N$, $\beta_{p,g}(\bv)$, and $\beta_{p,w}(\bv)$ are independent GPs with exponential kernels with parameters $(\sigma_{p,N}^2, \theta_{p,N})=(3, 6)$; $(\sigma_{p,g}^2, \theta_{p,g})= (3, 8)$ and $(\sigma_{p,w}^2, \theta_{p,w}) = (2, 5)$ respectively. Error variance is $\tau_p^2 = 0.01$.\\
& \\
2 & Misspecified network 
& Replace network effect $f_p^{N}(\tilde{\bu}_{i,p})$ in Scenario 1 with  $\sum_{j=1}^{R} \sin(\tilde{u}_{i,p,j})$. The varying coefficients will be simulated similar to Scenario 1.\\
&\\
3 & Heavy-tailed noise 
& Same as Scenario 1, but replace Gaussian noises with noises $\epsilon_{i,p}(\bv)$'s simulated from Student's $t_3$ distribution with the scale parameter chosen as $0.01/3$, such that the noise variance becomes $0.01$, same as Scenario 1. \\
&\\
4 & Non-additive interaction
& Simulate from $o_{i,p}(\bv)=
f_p^{N}(\tilde{\bu}_{i,p})
+ g_{i,p}(\bv)\beta_{p,g}(\bv)
+ w_{i,p}(\bv)\beta_{p,w}(\bv)\, +\rho f_p^{N}(\tilde{\bu}_{i,p}) \, g_{i,p}(\bv)\beta_{p,g}(\bv)
+ \epsilon_{i,p}(\bv)$, where $\epsilon_{i,p}(\bv) \sim N(0,\tau_p^{*2})$, $\tau_p^{*2}=0.01$, and $\rho=0.3$ controls the strength of the interaction. The functions $f_p^N$, $\beta_{p,g}(\bv)$, and $\beta_{p,w}(\bv)$ are simulated with an identical strategy as in Scenario 1.\\
\bottomrule
\end{tabular}
\end{table}

\subsection{Competitors \& Metrics of Comparison}
We evaluate two groups of competing models. The first group, \emph{nested competitors}, uses subsets of functional and network predictors to predict the outcome image, allowing us to assess each predictor's contribution to the model. The second group, \emph{non-nested competitors}, includes all three predictors like AMO-GP but does not account for network topology, spatial associations in functional predictors, or their structural interconnections.

\subsubsection{Competing Nested Models} \label{nested}

We consider two nested competitors that use subsets of predictors within the AMO-GP framework. In both cases, Stage 1 is identical to AMO-GP (Section~\ref{mf1}), yielding estimated latent node effects $\widehat{\bu}_{i,p}$, while modifications are introduced in Stage 2. The \emph{Network} model uses only the estimated network features $\widehat{\bu}_{i,p}$ to predict the response image $o_{i,p}(\bv)$ via a nonlinear regression, whereas the \emph{Function \& Network} model uses one functional predictor $g_{i,p}(\bv)$ and the estimated network features $\widehat{\bu}_{i,p}$ in the Stage 2 regression.

\subsubsection{Non-nested Competitors}\label{compete}

As non-nested competitors, we consider Bayesian additive regression trees (BART) \citep{chipman1998bayesian}, a neural network (NN), and a deep kernel learning Gaussian process model (BIRD-GP) \citep{ma2023}, each using all predictors to model the outcome image without imposing the additive multi-object structure of AMO-GP. BART is implemented using the \texttt{R} package \texttt{BART}, which fits a sum-of-trees model to capture nonlinear effects. The neural network is implemented using the \texttt{neuralnet} \texttt{R} package with a feedforward architecture. For BIRD-GP, we use the publicly available \texttt{Python} implementation provided by the authors, which employs deep kernel learning within a Gaussian process framework.

\subsubsection{Metrics of Comparison}\label{metric}
To compare AMO-GP with nested models, we use the posterior predictive loss criterion (PPLC) \citep{gelfand1998model}. Let $\widehat{o}_{i,p}(\bv)$ be the predicted response and $\widehat{\sigma}_{i,p}^2(\bv)$ the estimated variance. We use $G(\mathcal{M})=\sum_{p=1}^{P}\sum_{i=1}^{n}\sum_{\bv \in \mathcal{R}_p}^{}\left ( o_{i,p}(\bv)-\widehat{o}_{i,p}(\bv) \right )^{2}$ and $P(\mathcal{M})=\sum_{p=1}^{P}\sum_{i=1}^{n}\sum_{\bv \in \mathcal{R}_p}^{}\widehat{\sigma^2}_{i,p}(\bv)$ to compute accuracy and complexity, respectively, with total loss $D(\mathcal{M})=G(\mathcal{M})+P(\mathcal{M})$. We report $G(\mathcal{M})$, $P(\mathcal{M})$, $D(\mathcal{M})$ on training data along with the coefficient of determination ($R^2$) computed on the training set. For out-of-sample evaluation, we use Mean Squared Prediction Error (MSPE) on a test set of $n^*=100$ samples and report the corresponding test $R^2$. Predictive uncertainty is evaluated via coverage and length of 95\% intervals, averaged over all ROIs. NN provides only $G(\mathcal{M})$ and MSPE,  along with train and test $R^2$, without model complexity or predictive intervals. In addition, we assess the accuracy of coefficient estimation for AMO-GP through a coefficient recovery study, with details and full results provided in Section~\ref{sup_coef_recovery} of the Supplementary Materials.

\subsection{Simulation Results}
Table~\ref{sim_res} shows that AMO-GP consistently outperforms the \emph{Network} and \emph{Function \& Network} models in both accuracy and complexity across scenarios 1--3 — expected given the true model uses all three predictors. In particular, AMO-GP achieves substantially lower values of $G(\mathcal{M})$, $P(\mathcal{M})$, and $D(\mathcal{M})$, indicating both improved predictive fit and reduced effective model complexity. While BART, NN and BIRD-GP account for non-linear effects, they still substantially underperform relative to AMO-GP, with NN faring worse than all others. Scenario 4 marginally favors BIRD-GP over AMO-GP, however they both provide similar performance in this scenario. These results highlight two points: the two-stage approach effectively captures network effects, and performance degrades in Scenario 3 and 4 due to misspecified noise and increased complexity, respectively.
\begin{table}[!t]
\centering
\caption{\footnotesize{Average training and test performance of AMO-GP and competing methods over 10 repetitions across all scenarios. For the training data, we report $G(\mathcal{M})$, $P(\mathcal{M})$, $D(\mathcal{M})=G(\mathcal{M})+P(\mathcal{M})$, and $R^2$. For the test data, we report MSPE and $R^2$. $P(\mathcal{M})$ and $D(\mathcal{M})$ are unavailable for Neural Network (NN) and reported as ``NA''. Nested 1 includes one functional and network effects, whereas Nested 2 includes only network effects. Boldface indicates the best-performing method in each row.}}
\label{sim_res}
\small
\begin{tabular}{lcccccc}
\toprule
Metric & AMO-GP & \makecell{Function \\ \& Network} & Network & BART & NN & BIRD-GP \\
\midrule

\multicolumn{7}{c}{\textbf{Scenario 1}} \\
\midrule
$G(\mathcal{M}) \times 10^{-3}$ & \textbf{0.54} & 32.85 & 88.28 & 127.08 & 146.34 & 44.43 \\
$P(\mathcal{M}) \times 10^{-3}$ & \textbf{0.07} & 3.65 & 6.87 & 2.30 & NA & 2.32 \\
$D(\mathcal{M}) \times 10^{-3}$ & \textbf{0.61} & 36.50 & 95.14 & 129.38 & NA & 46.75 \\
$R^2$ (Train)    & \textbf{1.00} & 0.78 & 0.41 & 0.16 & 0.03 & 0.73 \\
MSPE (Test)      & \textbf{3.06} & 4.81 & 5.73 & 6.01 & 8.25 & 5.36 \\
$R^2$ (Test)     & \textbf{0.60} & 0.38 & -0.01 & -0.13 & -0.02 & 0.41 \\

\midrule
\multicolumn{7}{c}{\textbf{Scenario 2}} \\
\midrule
$G(\mathcal{M}) \times 10^{-3}$ & \textbf{0.54} & 106.75 & 280.85 & 341.43 & 380.90 & 132.84 \\
$P(\mathcal{M}) \times 10^{-3}$ & \textbf{0.07} & 9.53 & 15.75 & 4.55 & NA & 10.48 \\
$D(\mathcal{M}) \times 10^{-3}$ & \textbf{0.61} & 116.28 & 296.60 & 345.98 & NA & 143.32 \\
$R^2$ (Train)    & \textbf{1.00} & 0.75 & 0.38 & 0.16 & 0.09 & 0.64 \\
MSPE (Test)      & \textbf{1.34} & 3.56 & 6.32 & 6.57 & 8.59 & 6.31 \\
$R^2$ (Test)     & \textbf{0.77} & 0.33 & -0.00 & -0.23 & 0.00 & 0.34 \\

\midrule
\multicolumn{7}{c}{\textbf{Scenario 3}} \\
\midrule
$G(\mathcal{M}) \times 10^{-3}$ & \textbf{0.52} & 280.22 & 445.32 & 610.95 & 680.43 & 314.88 \\
$P(\mathcal{M}) \times 10^{-3}$ & \textbf{0.07} & 24.32 & 29.71 & 7.61 & NA & 29.30 \\
$D(\mathcal{M}) \times 10^{-3}$ & \textbf{0.59} & 304.54 & 475.03 & 618.56 & NA & 344.18 \\
$R^2$ (Train)    & \textbf{1.00} & 0.62 & 0.33 & 0.12 & 0.07 & 0.56 \\
MSPE (Test)      & \textbf{2.93} & 4.59 & 4.93 & 5.02 & 6.72 & 4.60 \\
$R^2$ (Test)     & \textbf{0.62} & 0.31 & -0.00 & -0.12 & -0.01 & 0.46 \\

\midrule
\multicolumn{7}{c}{\textbf{Scenario 4}} \\
\midrule
$G(\mathcal{M}) \times 10^{-3}$ & 44.18 & 46.18 & 104.18 & 136.90 & 165.67 & \textbf{39.64} \\
$P(\mathcal{M}) \times 10^{-3}$ & 5.16 & 4.96 & 7.74 & \textbf{2.98} & NA & 3.76 \\
$D(\mathcal{M}) \times 10^{-3}$ & 49.34 & 51.14 & 111.92 & 139.88 & NA & \textbf{43.40} \\
$R^2$ (Train)    & 0.92 & 0.72 & 0.37 & 0.18 & 0.01 & \textbf{0.95} \\
MSPE (Test)      & 4.02 & 6.75 & 9.31 & 9.75 & 10.99 & \textbf{3.75} \\
$R^2$ (Test)     & 0.55 & 0.34 & -0.01 & -0.17 & -0.01 & \textbf{0.56} \\

\bottomrule
\end{tabular}
\end{table}

Figure~\ref{covlen} summarizes predictive uncertainty via 95\% interval coverage and length, evaluated relative to the nominal 95\% level. All models exhibit under-coverage on training data, with this effect most pronounced for AMO-GP, which produces very narrow intervals reflecting high model flexibility. This is likely due to optimizing the weakly identifiable kernel parameters in GP, as argued by earlier articles encountering similar phenomena \citep{van2017convolutional}. On test data, AMO-GP achieves the highest coverage across scenarios ($\geq$ 90\%), while BART severely under-covers and other methods provide moderate but sub-nominal coverage.
\begin{figure}[!htp]
\centering
\includegraphics[width=0.9\textwidth]{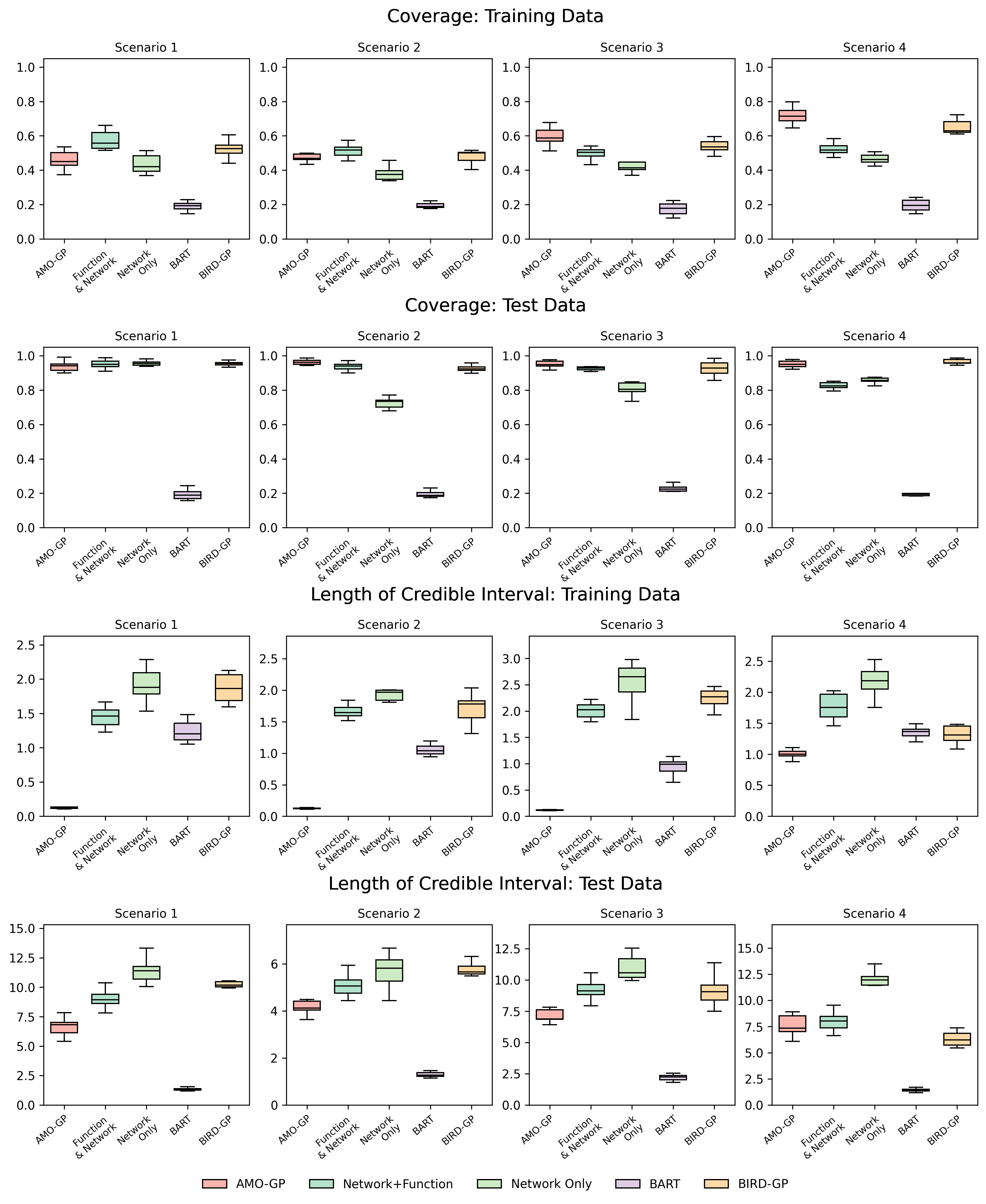}
\caption{\footnotesize{{Box plots of 95\% predictive interval coverage and length across all scenarios over 10 independent repetitions for AMO-GP and different competing models. Coverage and length are not available for NN.}}}
\label{covlen}
\end{figure}

In Scenario 3, most models (except BART) produce wider PIs than in Scenarios 1 and 2, reflecting sensitivity to error misspecification. BART yields the narrowest intervals across scenarios but suffers from severe under-coverage, indicating poor calibration. Incorporating both functional and network predictors improves both predictive accuracy and interval calibration compared to using network information alone. Additional coefficient recovery results (Supplementary Section~\ref{sup_coef_recovery}) show decreasing accuracy under misspecification, especially with heavy-tailed noise and interactions.

\section{ABCD Data Analysis}
\label{multi_modal_data_analysis}

As detailed in Section~\ref{data application}, the scientific goal is to predict task-based activation maps using structural MRI (cortical thickness and sulcal depth) and resting-state connectivity network from fMRI. 
The outcome variable from the t-fMRI data, and the explanatory variables, representing cortical thickness (CT) and sulcal depth (SD), are appropriately scaled and normalized to facilitate the analysis and generate results for inference. While AMO-GP regression flexibly models relationships across modalities, it is computationally demanding for large samples (see Section \ref{posterior_comp}). Thus, we analyzed neuroimaging data from 80 randomly selected ABCD study subjects, split equally into training and test sets (Section~\ref{data application}). 
Supplementary Section~\ref{sup-subsec-ABCDtraceplots} shows trace-plots for the length-scale and variance parameters, confirming satisfactory convergence. Model performance is compared to nested and non-nested alternatives described in Sections \ref{nested} and \ref{compete}. Nested models include: (i) connectivity network only predictor, (ii) connectivity network and SD predictors, and (iii) connectivity network and CT predictors. Non-nested competitors are BART, NN, and BIRD-GP. To assess the robustness of AMO-GP, we conducted an additional analysis by generating 10 random subsamples of size equal to the original training set (i.e., $40$) and fitting AMO-GP to each. We computed the in-sample $R^2$ across all subsamples and subnetworks, with results summarized via boxplots in Supplementary Figure~\ref{r2_subsample_fig}. The detailed results provided in Supplementary Section~\ref{sup-subsec-insampleR2} gives insight into the stability of model performance across different splits.

\subsection{Analysis \& Results}\label{data_ana_res}

\begin{table}
\centering
\caption{\footnotesize The model fitting statistics, including MSE, MSPE, out-of-sample coverage, and the length of 95\% predictive intervals for t-FMRI. Definitions of these metrics can be found in Section \ref{metric}. Since coverage and length metrics are not available for the NN model, the corresponding cells are marked as ``NA.''}
\label{data_metric_tab}
\resizebox{\textwidth}{!}{%
\begin{tabular}{@{}cccccccc@{}}
\toprule
\multirow{2}{*}{Metric}
& Proposed Model
& \multicolumn{3}{c}{Nested Competitors}
& \multicolumn{3}{c}{Non-Nested Competitors} \\
\cmidrule(lr){2-2} \cmidrule(lr){3-5} \cmidrule(lr){6-8}
& $\mathrm{AMO\text{-}GP}$
& $\mathrm{CT}$ \& $\mathrm{Network}$
& $\mathrm{SD}$ \& $\mathrm{Network}$
& $\mathrm{Network}$
& $\mathrm{BART}$
& $\mathrm{NN}$
& $\mathrm{BIRD\text{-}GP}$\\
\midrule
MSE               & 0.37    & 0.40    & 0.41    & 0.42    & 0.37    & 0.66 & 0.37\\
MSPE              & 1.50    & 1.55    & 1.55    & 1.56    & 1.90    & 6.61 & 3.93\\
Coverage (out-of-sample)       & 88.34\% & 83.78\% & 84.07\% & 83.34\% & 50.38\% & NA & 44.24\%\\
Length of CI (out-of-sample)     & 2.03    & 1.90    & 1.90    & 1.85    & 1.18    & NA  &0.50\\
\bottomrule
\end{tabular}%
}
\end{table}

Table \ref{data_metric_tab} compares AMO-GP with nested and non-nested competitors on the ABCD data. AMO-GP, BART, and BIRD-GP achieve the lowest in-sample MSE, followed by the nested competitors and NN. However, BIRD-GP exhibits substantially worse out-of-sample performance, as reflected by its higher MSPE. AMO-GP outperforms both BART and BIRD-GP out-of-sample and shows modest improvement over the nested competitors, which perform comparably overall, suggesting limited additional contribution from CT and SD after accounting for the network predictor. Subnetwork-level comparisons (Table~\ref{networkmse} of Supplementary Section~\ref{sup-subsec-ABCD-MSE}) show modest MSPE improvements from including structural predictors in the auditory, default, retrosplenial temporal, and sensorimotor hand networks. In particular, the default and retrosplenial temporal networks have been implicated in working memory \citep{vann2009, owen2005}. AMO-GP also achieves the highest out-of-sample coverage ($88.34\%$), indicating well-calibrated predictive uncertainty, at the cost of wider but more plausible predictive intervals. In contrast, BIRD-GP exhibits substantially lower coverage ($44.24\%$) along with much shorter predictive intervals, suggesting underestimation of predictive uncertainty and overconfident predictions. BART similarly shows low coverage ($50.38\%$), though with moderately wider intervals than BIRD-GP. These results highlight the advantage of the AMO-GP framework in providing more reliable uncertainty quantification. 

Across subnetworks, $R^2$ (in-sample) values range from approximately $0.32$ (none) to $0.87$ (retrosplenial temporal), with moderate variability across subsamples as indicated by standard deviations between $0.04$ and $0.11$, suggesting stable and consistently performance across subsamples. A detailed analysis of $R^2$ is provided in Section~\ref{sup-subsec-insampleR2} of Supplementary Materials. 
\begin{figure}
 \includegraphics[width=.5\textwidth, trim=50 0 50 0, clip]{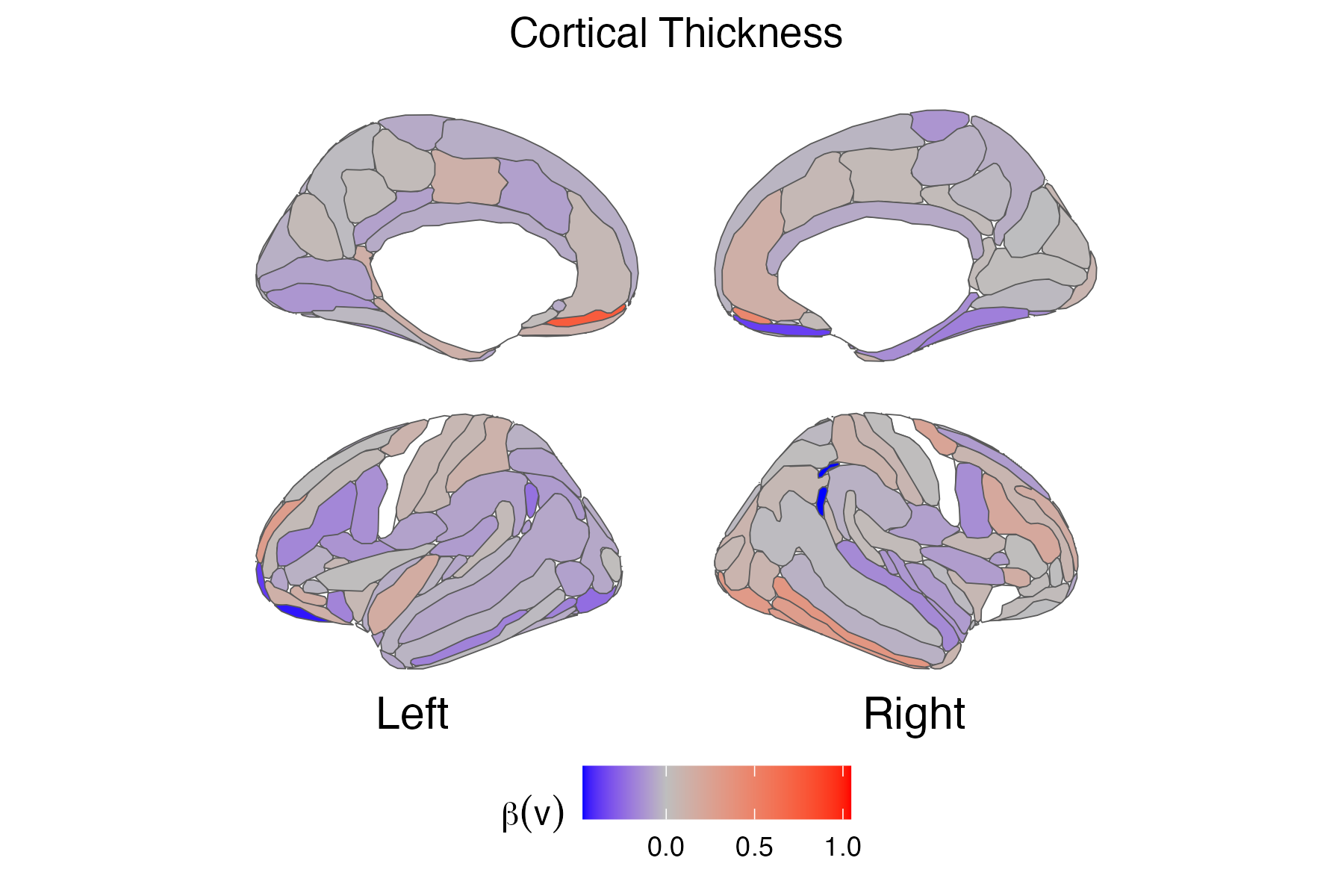}
 \includegraphics[width=.5\textwidth, trim=50 0 50 0, clip]{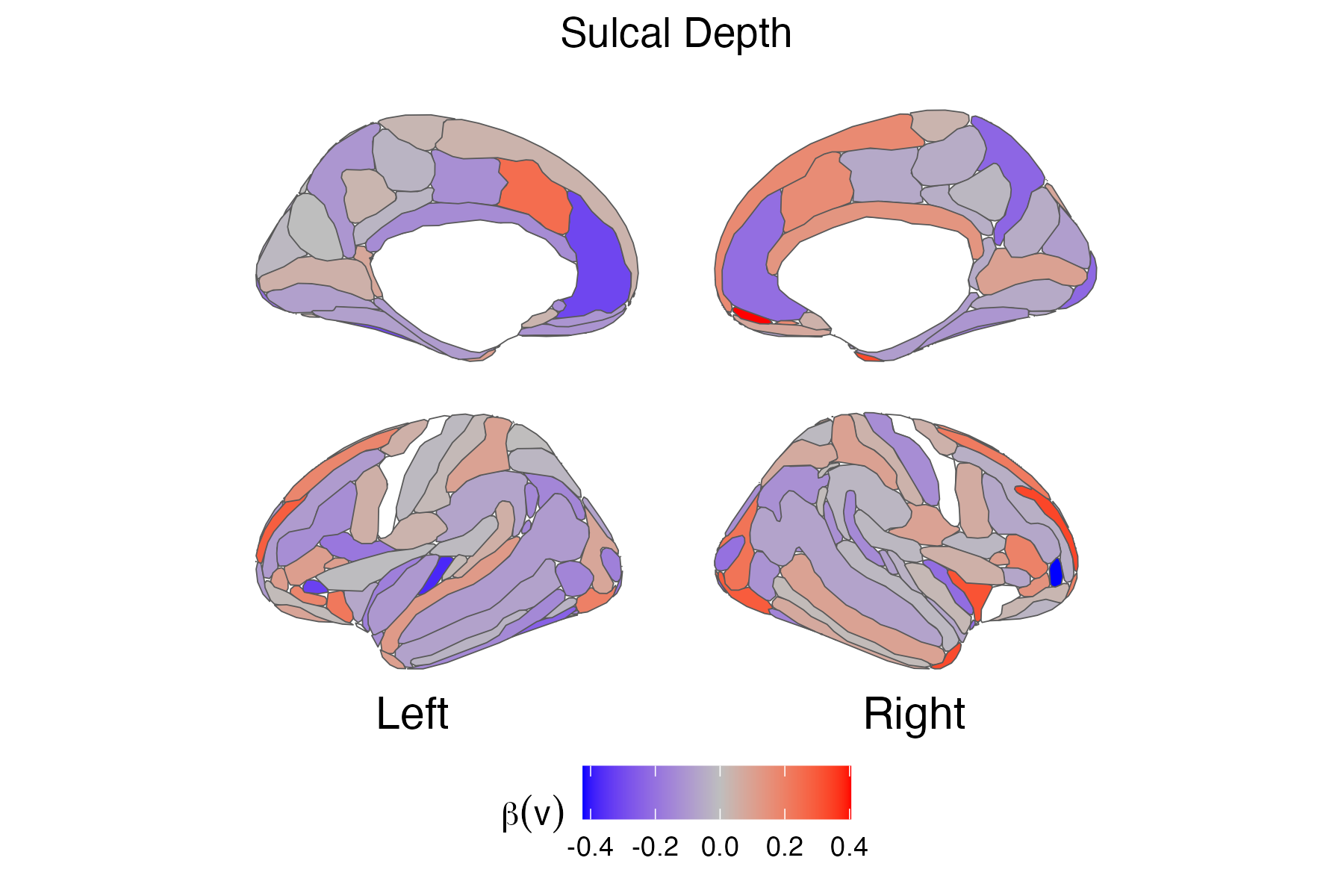}
\caption{\footnotesize{{Varying coefficient functions corresponding to the cortical thickness and sulcal depth s-MRI predictors. In each panel, the top row shows the medial view of the brain while the bottom row shows the lateral view of the brain. Similarly, the left column shows the left hemisphere while the right column shows the right hemisphere. Each ROI in the Destrieux parcellation is indicated by solid lines with shading reflecting the value of the varying coefficient function in that region.}}}
\label{var_cov}
\end{figure}

Figure \ref{var_cov} displays the varying coefficient functions corresponding to the cortical thickness and sulcal depth s-MRI predictors. In each panel, the top row shows the medial view of the brain and the bottom row the lateral view; the left and right columns correspond to the left and right hemispheres, respectively. Each Destrieux ROI is indicated by solid lines with shading reflecting the value of the varying coefficient function in that region which can be interpreted as the conditional effect of the predictor on the t-fMRI outcome (both outcome and predictors are scaled to have a standard deviation equal to one). Cortical thickness shows a generally negative association with task-based response to a working memory task while controlling for rs-fMRI and sulcal depth, particularly in the parietal region, which could reflect prior evidence that during development working memory is associated with a thinner parietal cortex \citep{krogsrud2021}. The exception is the right temporal lobe in which task activation is positively correlated with cortical thickness, an observation which is not well triangulated within the existing literature on working memory. Associations between sulcal depth and task-activation are more mixed: positive associations in the frontal and insula lobes and negative associations in the parietal, temporal, and occipital lobes. The limited work examining sulcal depth as a biomarker of working memory in the prefrontal cortex suggest similarly mixed findings \citep{yao2023}. The varying coefficient functions in Figure \ref{var_cov} highlight new inferential possibilities enabled by AMO-GP.

\section{Conclusion and Future Work}\label{conclusion}

This article introduces AMO-GP, a Bayesian framework for predictive inference of spatially-indexed functional outcomes using both functional and network image predictors. Additive nonlinear effects are modeled with Gaussian process priors on spatial coefficients for functional predictors and node effects for network predictors. AMO-GP enables uncertainty-quantified prediction of brain activation maps from t-fMRI based on cortical thickness, sulcal depth (s-MRI), and brain connectivity networks (rs-fMRI), establishing one of the first image-on-image regression models to incorporate both functional and network predictors.
In the ABCD data analysis, AMO-GP assessed the relative contributions of network and structural images in predicting task-based responses and provided interpretable conditional effects. Two limitations are noted: first, a limited sample size due to computational constraints necessitates cautious interpretation, as is common in neuroimaging studies. Second, the selection and alignment of brain atlases are critical; mismatches—such as structural ROIs spanning multiple subnetworks—may affect prediction accuracy, although inference and interpretation remain clear even when ideal nesting is not achieved. With further resources and access to the raw imaging data from the ABCD study, it would be possible to reprocess the imaging data to a single atlas more suited for joint analysis of structural and functional images, e.g. \citep{glasser}, and repeat the analysis, although this lies beyond the scope of this work. While the GP prior offers modeling flexibility, it has computational challenges with large datasets. A promising next step is to develop a Bayesian framework for efficient, distributed GP computation \citep{guhaniyogi2020large, guhaniyogi2023distributed}. 

\section*{Acknowledgments}
All computations were done using the \texttt{R} programming language on the Department of Statistics Arseven Computing Cluster at Texas A\&M University.

\subsection*{ABCD Data Acknowledgment}

Data used here were obtained from the Adolescent Brain Cognitive Development\textsuperscript{\texttrademark} (ABCD) Study (\url{https://abcdstudy.org}), held in the NIMH Data Archive (NDA). This is a multisite, longitudinal study designed to recruit more than 10,000 children age 9--10 and follow them over 10 years into early adulthood. The ABCD Study is supported by the National Institutes of Health and additional federal partners under award numbers U01DA041048, U01DA050989, U01DA051016, U01DA041022, U01DA051018, U01DA051037, U01DA050987, U01DA041174, U01DA041106, U01DA041117, U01DA041028, U01DA041134, U01DA050988, U01DA051039, U01DA041156, U01DA041025, U01DA041120, U01DA051038, U01DA041148, U01DA041093, U01DA041089, U24DA041123, U24DA041147. A full list of supporters is available at \url{https://abcdstudy.org/federal-partners.html}. A listing of participating sites and a complete listing of the study investigators can be found at \url{https://abcdstudy.org/consortium_members/}. 
ABCD consortium investigators designed and implemented the study and/or provided data but did not necessarily participate in the analysis or writing of this report. This manuscript reflects the views of the authors and may not reflect the opinions or views of the NIH or ABCD consortium investigators.

\section*{Supplementary Materials}
Supplementary material provides theoretical results on posterior consistency (Section~\ref{posterior_consis}), computation details (Section~\ref{posterior-deriv}), additional details and results on simulations (Section~\ref{supp:simdatgen}), and additional empirical results on real data analysis (Section~\ref{additional-data-study}). An \texttt{R} implementation of AMO-GP is available at \url{https://github.com/pritamdey/AMOGP}.

\section*{Data Availability}

The ABCD data repository grows and changes over time. The ABCD data used in this report came from \url{https://doi.org/10.15154/8873-zj65}.

\setstretch{1}
\putbib[bibliography_paper]

\end{bibunit}

\clearpage


\setcounter{section}{0}
\renewcommand{\thesection}{\Alph{section}}
\renewcommand{\thesubsection}{\thesection.\arabic{subsection}}
\renewcommand{\thesubsubsection}{\thesubsection.\arabic{subsubsection}}

\setcounter{figure}{0}
\setcounter{table}{0}
\setcounter{equation}{0}
\renewcommand{\thefigure}{S\arabic{figure}}
\renewcommand{\thetable}{S\arabic{table}}
\renewcommand{\theequation}{S\arabic{equation}}

\begin{bibunit}[plainnat]

\begin{center}
  {\LARGE\bfseries Supplementary Materials for\\
  \ManuscriptTitle\par}
  \vspace{0.2in}
\end{center}

\begin{abstract}
\noindent
This supplementary file consists of four sections. Section \ref{posterior_consis} introduces the notation and model setup and provides the statement and proof of the posterior consistency result in Theorem \ref{theorem1}. Section \ref{posterior-deriv} presents the full conditional distributions used to implement the Markov Chain Monte Carlo (MCMC) algorithm for model estimation. Section~\ref{supp:simdatgen} provides the detailed simulation settings for the simulation studies in Section~\ref{sim_studies} of the main manuscript. Finally, Section~\ref{additional-data-study} provides additional results from the ABCD data analysis discussed in Section~\ref{multi_modal_data_analysis} of the main text. 
\end{abstract}

\setstretch{1.4}

\section{Posterior Consistency of AMO-GP}\label{posterior_consis}

In this section, we prove the posterior consistency for the additive semi-parametric regression framework represented by equation~\eqref{additive_VCM} of the main article, assuming a single subject with only one network having $n$ ROIs such that the $D$-dimensional ROI co-ordinates take values in a compact set $\mathcal{T}\subset\mathbb{R}^D$. Without loss of generality, we assume $\mathcal{T}$ to be $[0,1]^D$. Since model represented by equation~\eqref{additive_VCM} of the main article is fitted in parallel for different subnetworks, we demonstrate posterior consistency for a specific subnetwork and omit the subscript $p$, denoting the functional covariates and varying coefficients as $g(\bv)$, $w(\bv)$, $\beta_g(\bv)$ and $\beta_w(\bv)$, respectively. Additionally, since we showcase our result for a single subject, we set $f_p^{N}=0$. Hence, $f(\bv)=g(\bv)\beta_g(\bv)+w(\bv)\beta_w(\bv)$ denotes the fitted regression function such that for $\bv_1,\hdots,\bv_n\in \mathcal{T}$ we have $O_i=f(\bv_i)+\epsilon_i$. Let us denote the sequences $O_1,\hdots,O_n$ and $f(\bv_1),\hdots,f(\bv_n)$ as $O_{1:n}$ and $f_n$ respectively. It is easy to check that $\bO_{1:n}|f \sim N(f_n,\tau^2\bI_n)$. For any given regression function $f$, we denote the conditional distribution and density of $\bO_{1:n}$ as $P_{n,f}$ and $p_{n,f}$ respectively.

Corresponding to the metric space $(\mathcal{T},||\cdot||)$, with $||\cdot||$ being the Euclidean norm, we denote by $C(\mathcal{T})$ the space of continuous functions on $\mathcal{T}$. Since $\mathcal{T}$ is compact with respect to the Euclidean norm, any continuous function $h$ on $\mathcal{T}$ is also bounded. We denote the supremum norm of any bounded function $h$ on $T$ by $ \left\| h\right\|_{\infty }=\text{sup}_{\bv \in \mathcal{T}}\left| h\left ( \bv \right )\right|$. For $\bv_1,\hdots,\bv_n\in \mathcal{T}$ and a function $h:\mathcal{T}\rightarrow \mathbb{R}$, the empirical norm $\left\| h\right\|_n$ is defined as
$\left\| h\right\|_n=\sqrt{\frac{1}{n}\sum_{i=1}^{n}h^{2}(\bv_i)}$. Further, let $\Pi$ denote the induced non-stationary mean-zero Gaussian process prior on the fitted regression function $f$ by assigning respective GP priors on varying coefficients as in equation~\eqref{eq:covariance_kernel} of the main article, and let $\left ( \mathbb{B},\left\| \cdot \right\|_{\mathbb{B}} \right )$ and $\left ( \mathbb{H},\left\| \cdot\right\|_{\mathbb{H}} \right )$ denote the respective separable Banach space and reproducing kernel Hilbert space (RKHS) corresponding to $\Pi$.

We will consider the following assumptions for showing posterior consistency.

\begin{enumerate}[label=(A.\arabic*), leftmargin=3em]
    \item \textbf{Fixed design.} 
    The co-ordinates $\bv_1,\hdots,\bv_n\in \mathcal{T}$ have been set by the experimenter.
    
    \item \textbf{Continuous functional covariates.} 
    The functional covariates $g(\cdot)$ and $w(\cdot)$ are non-zero and continuous on $\mathcal{T}$. 
    This implies that the support of the GP prior $\Pi$ is included in the space $C(\mathcal{T})$.

    \item \textbf{True regression model.} 
    The true regression function $f_0$ is included in $C(\mathcal{T})$.

    \item \textbf{Prior constraint.} 
    For any $\varepsilon > 0$, $\exists$ $n\left ( \varepsilon  \right )\geq 1/\varepsilon ^{2}$ such that 
    $\forall$ $n\geq n\left ( \varepsilon  \right )$ the following condition holds:
    \begin{align}
        \label{postcond}
        \Pi\left ( f:\left\| f\right\|_\infty <\varepsilon  \right )
        \geq 
        \Pi\left ( f:\left\| f-f_0\right\|_\infty <\varepsilon  \right )
        \geq 
        e^{-n\varepsilon ^{2}}.
    \end{align}
\end{enumerate}

Assumption (A.2) ensures that the GP $\Pi$ is non-degenerate and has a continuous bounded sample path. Assumption (A.3) ensures that the true regression function $f_0$ is contained in the closure of the RKHS $\mathbb{H}$ within $C(\mathcal{T})$. As shown in Lemma 5.1 of \citep{van2008reproducing}, the closure of $\mathbb{H}$ is the support of the prior $\Pi$ and in our case is the full space $C(\mathcal{T})$. (A.4) will be required to satisfy certain conditions while proving Theorem \ref{theorem1} for posterior consistency as stated below:
\begin{theorem}\label{theorem1}
Suppose the assumptions (A.1)-(A.4) hold and $P_{n,f_0}$ denotes the true data generating distribution. Then for every $\varepsilon>0$, 
\begin{align}\label{postthm}
P_{n,f_0}\left [ \Pi\left ( f:\left\| f-f_0\right\|_n>\varepsilon|\bO_{1:n} \right ) \right ]\to 0 \text{ as }n\to \infty.
\end{align}  
\end{theorem}

To prove Theorem \ref{theorem1}, we first need to develop some lemmata. Let $(V,d)$ be a normed space. Then for some $\Theta \subset V$, we denote its corresponding $\delta$-packing and $\delta$-covering numbers as $M(\delta,\Theta,d)$ and $N(\delta,\Theta,d)$ respectively. The following is a standard topological result to be applied in proving Lemma \ref{lemma1} stated after it.

\noindent\begin{result}\label{paco}
For any $\delta>0$ we have $M(2\delta,\Theta,d)\leq N(\delta,\Theta,d)\leq M(\delta,\Theta,d)$.
\end{result}

\noindent\begin{lemma}\label{lemma1}
For any $\delta>0$, suppose that assumption (A.4) is satisfied. Then for every $c>1$, $\delta>0$ and $n\geq n\left ( \delta  \right )$, $\exists$ a set $\Theta_{n,\delta,c}$ such that
\begin{gather}\label{setpr}
\Pi\left ( \Theta_{n,\delta ,c} \right ) \geq 1-e^{-2n\delta^{2}c^{2} } \\
\label{entpy}
M(6\sqrt{2}\delta, \Theta_{n,\delta ,c},\left\| \cdot\right\|_n ) \leq e^{12n\delta^{2}c^{2}} 
\end{gather}
\end{lemma}

\begin{proof}

Let $Z$ be the Borel measurable mean-zero Gaussian random element in separable Banach space $\mathbb{B}$ with RKHS $\mathbb{H}$ corresponding to $\Pi$. In addition, we denote $\mathbb{B}_1$ and $\mathbb{H}_1$ to be unit balls in $\mathbb{B}$ and $\mathbb{H}$ respectively. For any $c>1$, $\delta>0$ and $n\geq n\left ( \delta  \right )$, set a positive constant 
\begin{align}\label{mn}
 K_{n,\delta,c}=-2\Phi^{-1}( e^{-2n\delta^{2}c^{2}})    
\end{align}
where $\Phi$ denotes the distribution function of the standard normal distribution. Suppose that we have
\begin{align}\label{thetdef}
 \Theta_{n,\delta,c}= \delta_1\mathbb{B}_1 + K_{n,\delta,c}\mathbb{H}_1
\end{align}
where $\delta_1=\sqrt{2}\delta$. We set $\eta_{\delta}$ as the solution to the equation
\begin{align}\label{eta}
\Phi (\eta_{\delta} )=P\left ( Z \in \delta_1\mathbb{B}_1 \right )=\Pi\left ( f:\left\| f\right\|_\infty <\delta_1  \right ).    
\end{align}
Then, based on assumption (A.4) it is apparent that
$$\Phi (\eta_{\delta} ) \geq e^{-2n\delta^{2}} \geq e^{-2n\delta^{2}c^{2}}.$$

Therefore, we have $\eta_{\delta} + K_{n,\delta,c} \geq -\Phi^{-1}( e^{-2n\delta^{2}c^{2}})$. A simple application of Borell's inequality (see Theorem 3.1 in \citep{borell1975brunn} or Theorem 5.1 in \citep{van2008reproducing}) shows that
\begin{align*}
  \Pi\left ( \Theta_{n,\delta ,c} \right ) \geq  \Phi(\eta_{\delta} + K_{n,\delta,c}) \geq1-e^{-2n\delta^{2}c^{2} }.  
\end{align*}

Thus, we have proved (\ref{setpr}).

Suppose $\gamma_1,\hdots,\gamma_m$ are a maximal collection of points in $K_{n,\delta,c}\mathbb{H}_1$ such that they are $2\delta_1$ separated with respect to the norm $\left\| \cdot\right\|_{\mathbb{B}}$. Then we have $m=M(2\delta_1,K_{n,\delta,c}\mathbb{H}_1,\left\| \cdot\right\|_{\mathbb{B}})$. One may observe that, for $j=1,\hdots,m$, the $\left\| \cdot\right\|_{\mathbb{B}}$-balls $\gamma_j + \delta_1\mathbb{B}_1$ of radius $\delta_1$ are mutually disjoint. Further, applying (4.16) of \citep{kuelbs1994gaussian} or Lemma 5.2 of \citep{van2008reproducing} and under assumption (A.4), we can show that
\begin{align}\label{ball1}
    P\left ( Z \in \gamma_j + \delta_1\mathbb{B}_1 \right ) \geq e^{-\tfrac{1}{2}\left\| \gamma_j \right\|_{\mathbb{H}}^{2}}\Pi\left ( f:\left\| f\right\|_\infty <\delta_1  \right ) \geq e^{-\tfrac{1}{2}K_{n,\delta,c}^{2}}e^{-2n\delta^{2}c^{2}}.
\end{align}
Further, note that
\begin{align}\label{ball2}
  1\geq P\left (Z \in \bigcup_{j=1}^{m}(\gamma_j + \delta_1\mathbb{B}_1) \right )=\sum_{j=1}^{m} P\left ( Z \in \gamma_j + \delta_1\mathbb{B}_1 \right ).  
\end{align}
Combining (\ref{ball1}) and (\ref{ball2}) and applying Result \ref{paco} we can show that
\begin{align}\label{paco1}
N(2\delta_1,K_{n,\delta,c}\mathbb{H}_1,\left\| \cdot\right\|_{\mathbb{B}}) \leq m \leq e^{\tfrac{1}{2}K_{n,\delta,c}^{2}}e^{2n\delta^{2}c^{2}}.
\end{align}
Since any point in $\Theta_{n,\delta ,c}$ is within $\delta_1$ distance of some point in $K_{n,\delta,c}\mathbb{H}_1$ corresponding to the norm  $\left\| \cdot\right\|_{\mathbb{B}}$, we can show using triangle inequality that any collection of points which is a $2\delta_1$-cover for $K_{n,\delta,c}\mathbb{H}_1$ is also a $3\delta_1$-cover for $\Theta_{n,\delta ,c}$. This implies that
\begin{align}\label{paco2}
N(3\delta_1,\Theta_{n,\delta ,c},\left\| \cdot\right\|_{\mathbb{B}}) \leq N(2\delta_1,K_{n,\delta,c}\mathbb{H}_1,\left\| \cdot\right\|_{\mathbb{B}})  \leq e^{\tfrac{1}{2}K_{n,\delta,c}^{2}}e^{2n\delta^{2}c^{2}}.  
\end{align}
As has been reasoned in the proof of Theorem 2.1 in \citep{van2008rates}, we can show that $K_{n,\delta,c}\leq 20n\delta^{2}c^{2}$. Applying this result and Result \ref{paco} to (\ref{paco2}) we get
\begin{align}\label{paco3}
 M(6\delta_1,\Theta_{n,\delta ,c},\left\| \cdot\right\|_{\mathbb{B}}) \leq N(3\delta_1,\Theta_{n,\delta ,c},\left\| \cdot\right\|_{\mathbb{B}})\leq e^{12n\delta^{2}c^{2}}.   
\end{align}
Since our Banach space is $\left ( C(T),\left\|\cdot \right\|_{\infty } \right )$ and $\left\|\cdot \right\|_{n } \leq \left\|\cdot \right\|_{\infty }$ we have proved (\ref{entpy}).
\end{proof}

We have already discussed that $\bO_{1:n}|f \sim N(f_n,\tau^2\bI_n)$. To prove posterior consistency since we consider $\tau$ to be given, henceforth without loss of generality, we assume $\tau=1$.
\begin{lemma}\label{lemma2}
For any   $\delta>0$, $n\geq n\left ( \delta  \right )$  and $c \geq 6\sqrt{2}/5$, define a set $\Theta_{n,\delta,c}$ as in Lemma \ref{lemma1}. Then, $\exists$ a test $\psi$ based on $O_{1:n}|f \sim N(f_n,\bI_n)$ such that for every integer $j\geq1$,
\begin{gather}
\label{test1}
    P_{n,f_0}(\psi)\leq 9 e^{-\tfrac{1}{2}n\delta^{2}c^{2}},\\
\label{test2}
\sup_{f\in\Theta_{n,\delta,c}:\left\|f-f_0 \right\|_{n}\geq 10j\delta c}    P_{n,f}(1-\psi)\leq  e^{-100nj^2\delta^{2}c^{2}/8}
\end{gather}
where $f_0$ is the true response function.
\end{lemma}

\begin{proof}
Let us consider $\kappa>1$ and the test $\psi$ discussed in Lemma 13 of \citep{van2011information} and apply it in our present scenario corresponding to the set $\Theta_{n,\delta,c}$. Then we have for every integer $j\geq1$,
\begin{gather*}
 P_{n,f_0}(\psi) \leq 9M(\kappa/2,\Theta_{n,\delta ,c},\left\| \cdot\right\|)e^{-\kappa^{2}/8},\\
 \sup_{f\in\Theta_{n,\delta,c}:\left\|f-f_0 \right\|\geq j\kappa}    P_{n,f}(1-\psi) \leq e^{-j^2\kappa^{2}/8}
\end{gather*} 
where $\left\| \cdot\right\|$ is the Euclidean norm. Replacing $\kappa$ with $10\sqrt{n}\delta c$ we get for every integer $j\geq1$,
\begin{gather*}
 P_{n,f_0}(\psi)\leq 9M(5c\delta,\Theta_{n,\delta ,c},\left\| \cdot\right\|_n)e^{-100n\delta^{2}c^{2}/8},\\
 \sup_{f\in\Theta_{n,\delta,c}:\left\|f-f_0 \right\|_{n}\geq 10j\delta c} P_{n,f}(1-\psi)\leq  e^{-100nj^2\delta^{2}c^{2}/8}.
\end{gather*}
Thus, we have proved (\ref{test2}). Since we assumed $5c\geq 6\sqrt{2}$, applying (\ref{entpy}) of Lemma \ref{lemma1}  we get
\begin{align*}
   M(5c\delta, \Theta_{n,\delta ,c},\left\| \cdot\right\|_n ) \leq e^{12n\delta^{2}c^{2}} .
\end{align*}
Thus, we have proved (\ref{test1}).
\end{proof}

Finally, we are ready to provide the proof of Theorem \ref{theorem1}.

\subsection{Proof of Theorem \ref{theorem1}}

For any   $\delta>0$, $n\geq n\left ( \delta  \right )$  and $c \geq 6\sqrt{2}/5$, define a set $\Theta_{n,\delta,c}$ as in Lemma \ref{lemma1}. Further, consider the test $\psi$ based on $O_{1:n}|f \sim N(f_n,\bI_n)$ discussed in Lemma \ref{lemma2}. Define two events $\mathfrak{A}$ and $\mathfrak{B}$ such that in the event $\mathfrak{A}$,
\begin{align}\label{eventa}
 \int \frac{p_{n,f}}{p_{n,f_0}}d\Pi(f)\geq  e^{-n\delta^{2}c^{2}}\Pi(f:\left\| f-f_0\right\|_n < c\delta)   
\end{align}
while in the event $\mathfrak{B}$,
\begin{align}\label{eventb}
   \int \frac{p_{n,f}}{p_{n,f_0}}d\Pi(f)\geq  e^{-n\delta^{2}(c^{2}+1)}. 
\end{align}
Using assumption (A.4) we can show that 
\begin{align*}
e^{-n\delta^{2}c^{2}}\Pi(f:\left\| f-f_0\right\|_n < c\delta) \geq e^{-n\delta^{2}c^{2}}\Pi(f:\left\| f-f_0\right\|_n < \delta) \geq  e^{-n\delta^{2}(c^{2}+1)}.
\end{align*} 
Then we can conclude that the event $\mathfrak{A}$ implies $\mathfrak{B}$. Applying Lemma 14 of \citep{van2011information} to the event $\mathfrak{A}$ we get
\begin{align}\label{eventc}
    P_{n,f_0}(\mathfrak{B})\geq  P_{n,f_0}(\mathfrak{A}) \geq 1- e^{-n\delta^{2}c^{2}/8}.
\end{align}
Let $1_{\mathfrak{B}}$ be the indicator function for the event $\mathfrak{B}$. We can show that (see Proposition 11 of \citep{van2011information})
\begin{align*}
 P_{n,f_0}\left [ \Pi(f:\left\| f-f_0\right\|_{n}>10c\delta |O_{1:n})  \right ] \leq E_1 +E_2 +E_3 +E_4   
\end{align*}
such that
\begin{gather*}
E_1 = P_{n,f_0}(\psi ), \\
E_2 =P_{n,f_0}(\mathfrak{B}^{c} ), \\
E_3 = P_{n,f_0}\left [ \Pi(f\notin \Theta_{n,\delta,c} |O_{1:n})1_{\mathfrak{B}}  \right ],\\
E_4 = P_{n,f_0}\left [ \Pi(f\in \Theta_{n,\delta,c}:\left\| f-f_0\right\|_{n}>10c\delta |O_{1:n})(1-\psi )1_{\mathfrak{B}} \right ].
\end{gather*}
Applying (\ref{test1}) of Lemma \ref{lemma2} we can show that
\begin{align}\label{bounde1}
 E_1\leq    9 e^{-\tfrac{1}{2}n\delta^{2}c^{2}}.
\end{align}
Applying (\ref{eventc}) we can show that
\begin{align}\label{bounde2}
 E_2\leq     e^{-n\delta^{2}c^{2}/8}.
\end{align}
Applying (17) of \citep{van2011information} and (\ref{setpr}) of Lemma \ref{lemma1} we can show that 
\begin{align}\label{bounde3}
 E_3\leq     e^{-n\delta^{2}(c^{2}-1)}.
\end{align}
Using (\ref{eventb}) we can show that for any set $S$
\begin{align}\label{eventd}
 \Pi(S|O_{1:n})1_{\mathfrak{B}} \leq e^{n\delta^{2}(c^{2}+1)} \int_{S} \frac{p_{n,f}}{p_{n,f_0}}d\Pi(f).   
\end{align}
Note that $P_{n,f_0}\left [ (p_{n,f}/p_{n,f_0})(1-\psi ) \right ]\leq P_{n,f}(1-\psi )$. For every integer $j\geq1$, define disjoint sets $\Theta_{j,n,\delta,c}$ such that
\begin{align*}
  \Theta_{j,n,\delta,c}=\left\{f \in \Theta_{n,\delta,c}:10j\delta c \leq \left\|f-f_0 \right\|_n < 10(j+1)\delta c \right\}.  
\end{align*}
Then for any $f \in \Theta_{j,n,\delta,c}$, applying (\ref{test2}) of Lemma \ref{lemma2}, we have $P_{n,f}(1-\psi)\leq  e^{-100nj^2\delta^{2}c^{2}/8}$. Then using (\ref{eventd}), Fubini's Theorem and (\ref{test2}) of Lemma \ref{lemma2} we have
\begin{align}\label{bounde4}
E_4 &\leq e^{n\delta^{2}(c^{2}+1)}\sum_{j\geq 1}P_{n,f_0}(1-\psi )\int_{ \Theta_{j,n,\delta,c}}\frac{p_{n,f}}{p_{n,f_0}}d\Pi(f) \leq e^{n\delta^{2}(c^{2}+1)}\sum_{j\geq 1}  e^{-100nj^2\delta^{2}c^{2}/8} \leq 9 e^{-n\delta^{2}(c^{2}-1)}
\end{align}
as $1/(1-e^{-100n\delta^{2}c^{2}/8}) \leq 1/(1-e^{-1/8})\leq 9$. Then combining (\ref{bounde1}), (\ref{bounde2}), (\ref{bounde3}) and (\ref{bounde4}) we get the following.
\begin{align*}
 P_{n,f_0}\left [ \Pi(f:\left\| f-f_0\right\|_{n}>10c\delta |O_{1:n})  \right ] \leq 20  e^{-n\delta^{2}(c^{2}-1)/8}. 
\end{align*}
For any $\varepsilon>0$, choose some $c \geq 6\sqrt{2}/5$ and set $\delta=\varepsilon/10c$. Then for every $n\geq n\left ( \delta  \right )$ we have
\begin{align*}
 P_{n,f_0}\left [ \Pi(f:\left\| f-f_0\right\|_{n}>\varepsilon |O_{1:n})  \right ] \leq 20  e^{-n\varepsilon^{2}/1600}
\end{align*}
as$(c^2-1)/c^2 \geq 1/2$ for any $c \geq 6\sqrt{2}/5$. Then, as $n \to \infty$, we have the desired result. Hence the proof of Theorem \ref{theorem1} is complete.

\newpage
\section{Detailed Computation of Posteriors}\label{posterior-deriv}

\subsection{Conditional Posterior Distributions in Stage 1 of AMO-GP}

Aligning with the notations in Section~\ref{mf1} of the main manuscript, we have

$$
\bA_{i}^{P\times P}=\begin{pmatrix}
0 & a_{i,(1,2)} &\hdots   & a_{i,(1,P)} \\
a_{i,(1,2)} & 0 & \hdots & a_{i,(2,P)} \\
\vdots  & \vdots & \vdots & \vdots \\
 a_{i,(1,P)}& \hdots & a_{i,(P-1,P)}  & 0 \\
\end{pmatrix}.
$$

Now define

$$\ba_{i}^{\tfrac{P\left ( P-1 \right )}{2}\times 1}=\begin{pmatrix}
a_{i,(1,2)} \\
\vdots  \\
a_{i,(1,P)}\\
a_{i,(2,3)}\\
\vdots\\
a_{i,(2,P)}\\
\vdots\\
a_{i,(P-1,P)}
\end{pmatrix}, \;\bepsilon_{i}^{\tfrac{P\left ( P-1 \right )}{2}\times 1}=\begin{pmatrix}
\epsilon_{i,(1,2)} \\
\vdots  \\
\epsilon_{i,(1,P)}\\
\epsilon_{i,(2,3)}\\
\vdots\\
\epsilon_{i,(2,P)}\\
\vdots\\
\epsilon_{i,(P-1,P)}
\end{pmatrix} \;\text{ and }\bw_{i}^{\tfrac{P\left ( P-1 \right )}{2}\times 1}=\begin{pmatrix}
\tilde{\bu}_{i,1}^{T}\tilde{\bu}_{i,2} \\
\vdots  \\
\tilde{\bu}_{i,1}^{T}\tilde{\bu}_{i,P}\\
\tilde{\bu}_{i,2}^{T}\tilde{\bu}_{i,3}\\
\vdots\\
\tilde{\bu}_{i,2}^{T}\tilde{\bu}_{i,P}\\
\vdots\\
\tilde{\bu}_{i,P-1}^{T}\tilde{\bu}_{i,P}
\end{pmatrix}.
$$

Then we have

$$
\ba_{i}= \mu_i\mathbf{1}^{\tfrac{P\left ( P-1 \right )}{2}\times 1} + \bw_{i} + \bepsilon_{i}, \:\: i \in \left\{ 1,\hdots,n \right\}.
$$

Assuming a flat prior for $\mu_i$ $\left ( i \in \left\{ 1,\hdots,n \right\} \right )$, conditioning on other parameters we have the conditional posterior distribution of $\mu_i$ 

$$
f(\mu _{i}|\ba_{i},\bw_{i},\sigma_u^2)\propto N(\ba_{i}|\mu_i\mathbf{1}+\bw_{i},\sigma_u^2\boldsymbol{I}_{\tfrac{P\left ( P-1 \right )}{2}})\times 1.
$$
Then a few algebraic steps will show us that $\mu _{i}|\ba_{i},\bw_{i},\sigma_u^2 \sim N(\tfrac{2}{P\left ( P-1 \right )}\mathbf{1}^{T}(\ba_{i}-\bw_{i}),\tfrac{2}{P\left ( P-1 \right )}\sigma_u^2)$.
For the conditional posterior distribution of $\sigma_u^2$  we have,
$$
f(\sigma_u^2|\ba,\bw, \bmu)\propto N(\ba|\bmu+\bw,\sigma_u^2\boldsymbol{I}_{\tfrac{nP\left ( P-1 \right )}{2}})\times IG(\sigma_u^2|\alpha,\beta)
$$
where we assume $\alpha=\beta=1$ and define
$$\ba^{\tfrac{nP\left ( P-1 \right )}{2}\times 1}=\begin{pmatrix}
\ba_{1} \\
\vdots  \\
\ba_{n}
\end{pmatrix}, \;\bw^{\tfrac{nP\left ( P-1 \right )}{2}\times 1}=\begin{pmatrix}
\bw_{1} \\
\vdots  \\
\bw_{n}
\end{pmatrix} \;\text{ and }\bmu^{\tfrac{nP\left ( P-1 \right )}{2}\times 1}=\begin{pmatrix}
\mu_1\mathbf{1} \\
\vdots  \\
\mu_n\mathbf{1}
\end{pmatrix}.
$$
Then a few algebraic steps will show us that $\sigma_u^2|\ba,\bw, \bmu \sim IG(\alpha+\tfrac{nP(P-1)}{4},\beta + \tfrac{\left\| \ba - \bmu - \bw \right\|^{2}}{2})$, $\left\| \cdot \right\|$ representing the Euclidean norm.\\

We have $\forall \: p \in \left\{ 1,\hdots,P \right\}$ and $\forall \: i \in \left\{ 1,\hdots,n \right\}$,  $\tilde{\bu}_{i,p}^{R \times 1}=\begin{pmatrix}
u_{i,p,1} \\
\vdots  \\
u_{i,p,R}
\end{pmatrix}$ and define
$$
\ba_{i,p}^{(P-1) \times 1}=\begin{pmatrix}
a_{i,(p,1)} \\
\vdots  \\
a_{i,(p,p-1)}\\
a_{i,(p,p+1)}\\
\vdots \\
a_{i,(p,P)}
\end{pmatrix} , \;
\bU_{i,p}^{(P-1) \times R}=\begin{pmatrix}
\tilde{\bu}_{i,1}^{T} \\
\vdots  \\
\tilde{\bu}_{i,p-1}^{T}\\
\tilde{\bu}_{i,p+1}^{T}\\
\vdots \\
\tilde{\bu}_{i,P}^{T}
\end{pmatrix} \text{ and } \bepsilon_{i,p}^{(P-1) \times 1}=\begin{pmatrix}
\epsilon_{i,(p,1)} \\
\vdots  \\
\epsilon_{i,(p,p-1)}\\
\epsilon_{i,(p,p+1)}\\
\vdots \\
\epsilon_{i,(p,P)}
\end{pmatrix}.
$$

Then it is easy to see that 
$$
\ba_{i,p}=\mu_i\mathbf{1}^{(P-1)\times 1} + \bU_{i,p}\tilde{\bu}_{i,p} + \bepsilon_{i,p}.
$$

For the conditional posterior distribution of $\tilde{\bu}_{i,p}$ we have,
$$
f(\tilde{\bu}_{i,p}|\ba_{i,p},\mu_i,\bU_{i,p},\sigma_u^2) \propto N(\ba_{i,p}|\mu_i\mathbf{1}^{(P-1)\times 1} + \bU_{i,p}\tilde{\bu}_{i,p},\sigma_u^2\boldsymbol{I}_{P-1} ) \times N(\tilde{\bu}_{i,p}|\mathbf{0}^{R\times 1},\boldsymbol{I}_R)
$$

Then a few algebraic steps will show us that $
\tilde{\bu}_{i,p}|\ba_{i,p},\mu_i,\bU_{i,p},\sigma_u^2 \sim N(\bet_{i,p},\bPhi_{i,p})$ where

$$
\bPhi_{i,p}=(\boldsymbol{I}_R+\tfrac{1}{\sigma_u^2}\bU_{i,p}^T\bU_{i,p})^{-1} \text{ and } \bet_{i,p}=\tfrac{1}{\sigma_u^2}\bPhi_{i,p}\bU_{i,p}^T(\ba_{i,p}-\mu_i\mathbf{1}).
$$

\subsection{Conditional Posterior Distributions in Stage 2 of AMO-GP}

Aligning with the notations in Section~\ref{mf2} of the main manuscript, we define
$$
\bo_{i,p}^{V_p \times 1}= \begin{pmatrix}
o_{i,p}(\bv_{p,1}) \\
\vdots \\
o_{i,p}(\bv_{p,V_p}) 
\end{pmatrix}, \; \bbeta_{p,g}^{V_p \times 1}= \begin{pmatrix}
\beta_{p,g}(\bv_{p,1}) \\
\vdots \\
\beta_{p,g}(\bv_{p,V_p}) 
\end{pmatrix},$$
$$
\bbeta_{p,w}^{V_p \times 1}= \begin{pmatrix}
\beta_{p,w}(\bv_{p,1}) \\
\vdots \\
\beta_{p,w}(\bv_{p,V_p}) 
\end{pmatrix}, \; \bepsilon_{i,p}^{V_p \times 1}= \begin{pmatrix}
\epsilon_{i,p}(\bv_{p,1}) \\
\vdots \\
\epsilon_{i,p}(\bv_{p,V_p}) 
\end{pmatrix},
  $$
 $$
 \bG_{i,p}^{V_p \times V_p}=\begin{bmatrix}
    g_{i,p}(\bv_{p,1}) & & \\
    & \ddots & \\
    & & g_{i,p}(\bv_{p,V_p})
  \end{bmatrix} \; \text{ and } \bW_{i,p}^{V_p \times V_p}=\begin{bmatrix}
    w_{i,p}(\bv_{p,1}) & & \\
    & \ddots & \\
    & & w_{i,p}(\bv_{p,V_p})
  \end{bmatrix}.  
$$

where $\bv_{p,1},\hdots,\bv_{p,V_p}$ represent coordinates corresponding to the $V_p$ ROIs nested within the $p$th subnetwork. Assuming $n$ subjects corresponding to the training data, we denote

$$
\bo_{p}^{nV_p \times 1}= \begin{pmatrix}
\bo_{1,p} \\
\vdots \\
\bo_{n,p}
\end{pmatrix}, \; \mathbf{f}_{p}^{n \times 1}=\begin{pmatrix}
f_p^{N}(\widehat{\bu}_{1,p}) \\
\vdots \\
f_p^{N}(\widehat{\bu}_{n,p})
\end{pmatrix}, \; \bepsilon_{p}^{nV_p \times 1}= \begin{pmatrix}
 \bepsilon_{1,p} \\
\vdots \\
 \bepsilon_{n,p}
\end{pmatrix},
$$
$$
\bG_{p}^{nV_p \times V_p}= \begin{bmatrix}
\bG_{1,p} \\
 \vdots\\
 \bG_{n,p}
\end{bmatrix}, \; \bW_{p}^{nV_p \times V_p}= \begin{bmatrix}
\bW_{1,p} \\
 \vdots\\
 \bW_{n,p}
\end{bmatrix} \; \text{ and } \bL_p^{nV_p \times n}=\begin{bmatrix}
    \mathbf{1}^{V_p \times 1} & & \\
    & \ddots & \\
    & & \mathbf{1}^{V_p \times 1}
  \end{bmatrix}.
$$

Then we have,

$$
\bo_p=\bL_p\, \mathbf{f}_{p} + \bG_p\, \bbeta_{p,g} + \bW_p\, \bbeta_{p,w} + \bepsilon_{p}.
$$

Let $\Sigma_{p,g}$, $\Sigma_{p,w}$ and $\Sigma_{p,N}$ denote the covariance matrices corresponding to exponential covariance kernels $\kappa_{p,g}$, $\kappa_{p,w}$, and $\kappa_{p,N}$ for $\bbeta_{p,g}$, $\bbeta_{p,w}$ and $\mathbf{f}_{p}$ respectively. Further, let us define the parameter set $\Delta=\left\{ \tau_p^2, \theta_{p,g},\sigma_{p,g}^2, \theta_{p,w},\sigma_{p,w}^2, \theta_{p,N},\sigma_{p,N}^2\right\}$. Note that 
$$
\bepsilon_{p} \sim N(\bzero^{nV_p \times 1},\tau_p^2\boldsymbol{I}_{nV_p}), \; \bbeta_{p,g}|\theta_{p,g},\sigma_{p,g}^2 \sim N(\bzero,\Sigma_{p,g}),
$$
$$
\bbeta_{p,w}|\theta_{p,w},\sigma_{p,w}^2 \sim N(\bzero,\Sigma_{p,w}), \; \mathbf{f}_{p}|\theta_{p,N},\sigma_{p,N}^2 \sim N(\bzero,\Sigma_{p,N}) \text{ and} 
$$
$$
\bo_p|\Delta \sim N(\bzero, \bL_p\Sigma_{p,N}\bL_p^{T}+ \bG_p\Sigma_{p,g}\bG_p^{T} + \bW_p\Sigma_{p,w}\bW_p^{T} + \tau_p^2\boldsymbol{I}_{nV_p}).
$$

Let $\zeta_1$, $\zeta_2$, $\zeta_3$, $\zeta_4$ be the logarithmic transformations of the variance parameters $\tau_p^2$, $\sigma_{p,g}^2$, $\sigma_{p,w}^2$ and $\sigma_{p,N}^2$. To implement the Metropolis–Hastings algorithm we consider the joint conditional posterior distribution of $(\zeta_1$, $\zeta_2$, $\zeta_3$, $\zeta_4)$ such that
$$
f(\zeta_1, \zeta_2, \zeta_3, \zeta_4|\bo_p,\theta_{p,g},\theta_{p,w},\theta_{p,N}) \propto f(\bo_p|\Delta) \times f(\zeta_1, \zeta_2, \zeta_3, \zeta_4)
$$
We denote $\varsigma=\tau_p^2 \times \sigma_{p,g}^2 \times \sigma_{p,w}^2 \times \sigma_{p,N}^2$, assume $\alpha_1=\alpha_2=\alpha_3=\alpha_4=1$, $\gamma_1=\gamma_2=\gamma_3=\gamma_4=2$ and have
$$
f(\zeta_1, \zeta_2, \zeta_3, \zeta_4)=IG(\tau_p^2|\alpha_1,\gamma_1)\times IG(\sigma_{p,g}^2|\alpha_2,\gamma_2)\times IG(\sigma_{p,w}^2|\alpha_3,\gamma_3)\times IG(\sigma_{p,N}^2|\alpha_4,\gamma_4) \times \varsigma.
$$

For the conditional posterior distribution of $\mathbf{f}_{p}$ we have
$$
f(\mathbf{f}_{p}|\bo_p, \Delta) \propto N(\bo_p|\bL_p\mathbf{f}_{p}, \Sigma_{p,g,w}) \times N(\mathbf{f}_{p}|\bzero,\Sigma_{p,N})
$$

where $\Sigma_{p,g,w}=\bG_p\Sigma_{p,g}\bG_p^{T} + \bW_p\Sigma_{p,w}\bW_p^{T} + \tau_p^2\boldsymbol{I}_{nV_p}$. Define $\Sigma_{p,N,O}= (\Sigma_{p,N}^{-1} + \bL_p^{T}\Sigma_{p,g,w}^{-1}\bL_p )^{-1}$. Then a few algebraic steps show us that 

$$
\mathbf{f}_{p}|\bo_p, \Delta \sim N(\Sigma_{p,N,O}\bL_p^{T}\Sigma_{p,g,w}^{-1}\bo_p, \Sigma_{p,N,O}). 
$$

For the conditional posterior distribution of $\bbeta_{p,g}$ we have
$$
f(\bbeta_{p,g}|\bo_p, \Delta) \propto N(\bo_p|\bG_p\bbeta_{p,g}, \Sigma_{p,N,w}) \times N(\bbeta_{p,g}|\bzero,\Sigma_{p,g})
$$

where $\Sigma_{p,N,w}=\bL_p\Sigma_{p,N}\bL_p^{T} + \bW_p\Sigma_{p,w}\bW_p^{T} + \tau_p^2\boldsymbol{I}_{nV_p}$. Define $\Sigma_{p,g,O}= (\Sigma_{p,g}^{-1} + \bG_p^{T}\Sigma_{p,N,w}^{-1}\bG_p )^{-1}$. Then a few algebraic steps show us that 

$$
\bbeta_{p,g}|\bo_p, \Delta \sim N(\Sigma_{p,g,O}\bG_p^{T}\Sigma_{p,N,w}^{-1}\bo_p, \Sigma_{p,g,O}). 
$$

For the conditional posterior distribution of $\bbeta_{p,w}$ we have
$$
f(\bbeta_{p,w}|\bo_p, \Delta) \propto N(\bo_p|\bW_p\bbeta_{p,w}, \Sigma_{p,N,g}) \times N(\bbeta_{p,w}|\bzero,\Sigma_{p,w})
$$

where $\Sigma_{p,N,g}=\bL_p\Sigma_{p,N}\bL_p^{T} + \bG_p\Sigma_{p,g}\bG_p^{T} + \tau_p^2\boldsymbol{I}_{nV_p}$. Define $\Sigma_{p,w,O}= (\Sigma_{p,w}^{-1} + \bW_p^{T}\Sigma_{p,N,g}^{-1}\bW_p )^{-1}$. Then a few algebraic steps show us that 

$$
\bbeta_{p,w}|\bo_p, \Delta \sim N(\Sigma_{p,w,O}\bW_p^{T}\Sigma_{p,N,g}^{-1}\bo_p, \Sigma_{p,w,O}). 
$$

 Suppose we have $m$ subjects in the test data. Similar to $\mathbf{f}_{p}$ corresponding to the training data, let $\mathbf{f}_{t,p}^{m \times 1}$ represent the vector of non-linear functional effects of $\widehat{\bu}_{i,p}$ for $i=1,\hdots,m$ corresponding to $m$ subjects in the test data. Corresponding to the kernel $\kappa_{p,N}$, suppose we denote $\Sigma_{p,t}^{m \times m}=Var(\mathbf{f}_{t,p})$ and $\Sigma_{p,N,t}^{n \times m}=Cov(\mathbf{f}_{p},\mathbf{f}_{t,p})$. Further, let

$$
\mathbf{f}_{s,p}^{(n+m) \times 1}= \begin{pmatrix}
 \mathbf{f}_{p} \\
 \mathbf{f}_{t,p}
\end{pmatrix}\; \text{ and } \; \Sigma_{p,s}^{(n+m) \times (n+m)}=\begin{bmatrix}
\Sigma_{p,N} & \Sigma_{p,N,t} \\
\Sigma_{p,N,t}^{T} &  \Sigma_{p,t}\\
\end{bmatrix}.
$$ 

Note that $\mathbf{f}_{s,p} \sim N(\bzero,\Sigma_{p,s})$. Then a few algebraic steps show us that
$$
f(\mathbf{f}_{t,p}|\bo_p,\Delta)=N(\mathbf{f}_{t,p}|\Sigma_{p,N,t}^{T}\Sigma_{p,N}^{-1}\mathbf{f}_{p}, \Sigma_{p,t} - \Sigma_{p,N,t}^{T}\Sigma_{p,N}^{-1}\Sigma_{p,N,t}) \times f(\mathbf{f}_{p}|\bo_p,\Delta).
$$

\newpage
\section{Detailed Simulation Setup}\label{supp:simdatgen}

This section provides full details of the data-generating mechanisms used in the simulation studies described in Section~\ref{sim_studies} of the main text to ensure complete reproducibility.

\subsection{Network Model for Stage 1} 

We consider $P=50$ subnetworks and $n=100$ subjects for both training and test datasets. Each subnetwork contains $V=12$ regions of interest (ROIs), indexed by spatial coordinates $\bv=(v_1,v_2,v_3) \in [0,10]^3$. For each replicate, ROI locations are independently sampled from a uniform distribution over $[0,10]^3$ and are shared across subjects within each subnetwork.

The network predictor is constructed using $R=3$ latent components. For each subject $i=1,\ldots,n$ and subnetwork $p=1,\ldots,P$, we generate latent variables
\begin{align*}
\tilde{\bu}_{i,p} \sim N(\mathbf{0}, \mathbf{I}_3), \qquad
\mu_i \sim N(0,1), \qquad
\epsilon_{i,(p,p')} \sim N(0,0.01), \quad p \neq p',
\end{align*}
independently across all indices. The $(p,p')$th edge of the network predictor $a_{i,(p,p')}$ is then constructed according to:
\begin{align}\label{latent_scale-supp}
a_{i,(p,p')}=\mu_i+\tilde{\bu}_{i,p}^{T}\tilde{\bu}_{i,p'}+\epsilon_{i,(p,p')}.
\end{align}

\subsection{Functional Predictors for Stage 2}

For each subject $i$, subnetwork $p$, and ROI location $\bv$, we generate two functional predictors independently as 
$w_{i,p}(\bv) \sim N(0,1)$ and $g_{i,p}(\bv) \sim N(0,1)$,
independently across all $i,p,\bv$.

\subsection{Gaussian Process Specifications}

In scenarios where Gaussian process (GP) priors are used, we generate functions independently across subnetworks. Specifically, for each $p$, $f_p^N(\cdot)$ is drawn from a zero-mean GP with exponential covariance kernel $\kappa_{p,N}$; and 
$\beta_{p,g}(\bv)$ and $\beta_{p,w}(\bv)$ are drawn from zero-mean GPs with kernels $\kappa_{p,g}$ and $\kappa_{p,w}$, respectively. The covariance kernels of these GPs kernels are all exponential kernels, i.e., $\kappa(\bx, \bx') = \sigma^2 \exp\left(-\theta\|\bx - \bx'\|\right)$, with the kernel parameters fixed across subnetworks as follows:
\begin{align*}
\sigma_{p,N}^{*2} = 3, \qquad \theta_{p,N}^* = 6, \qquad
\sigma_{p,g}^{*2} = 3, \qquad \theta_{p,g}^* = 8, \qquad
\sigma_{p,w}^{*2} = 2, \qquad \theta_{p,w}^* = 5.
\end{align*}

\subsection{Scenario-Specific Data Generation}

For each subject $i$, subnetwork $p$, and ROI $\bv$, the response $o_{i,p}(\bv)$ is generated as follows.

\medskip
\noindent\underline{\textbf{Scenario 1 (Additive Gaussian process model).}}  
The response is generated following:
\begin{align*}
o_{i,p}(\bv)
=
f_p^{N}(\tilde{\bu}_{i,p})
+ g_{i,p}(\bv)\beta_{p,g}(\bv)
+ w_{i,p}(\bv)\beta_{p,w}(\bv)
+ \epsilon_{i,p}(\bv),
\end{align*}
where $\epsilon_{i,p}(\bv) \sim N(0,\tau_p^{*2})$ with $\tau_p^{*2}=0.01$.

\medskip
\noindent\underline{\textbf{Scenario 2 (Nonlinear network effect).}}  
The response is generated as
\begin{align*}
o_{i,p}(\bv)
=
\sum_{j=1}^{3} \sin(\tilde{u}_{i,p,j})
+ g_{i,p}(\bv)\beta_{p,g}(\bv)
+ w_{i,p}(\bv)\beta_{p,w}(\bv)
+ \epsilon_{i,p}(\bv),
\end{align*}
with $\epsilon_{i,p}(\bv) \sim N(0,\tau_p^{*2})$ and $\tau_p^{*2}=0.01$. The coefficient functions $\beta_{p,g}(\bv)$ and $\beta_{p,w}(\bv)$ are generated as in Scenario 1.

\medskip
\noindent\underline{\textbf{Scenario 3 (Heavy-tailed noise).}}  
The response is generated as in Scenario 1, but with heavy-tailed noise,
$\epsilon_{i,p}(\bv) \sim t_3 \cdot \sqrt{\tau_p^{*2}/\mathrm{Var}(t_3)}$,
where $\tau_p^{*2}=0.01$ ensures comparable marginal variance to the Gaussian setting.

\medskip
\noindent\underline{\textbf{Scenario 4 (Non-additive interaction).}}  
The response is generated as
\begin{align*}
o_{i,p}(\bv)=
f_p^{N}(\tilde{\bu}_{i,p})
+ g_{i,p}(\bv)\beta_{p,g}(\bv)
+ w_{i,p}(\bv)\beta_{p,w}(\bv)\, +\rho f_p^{N}(\tilde{\bu}_{i,p}) \, g_{i,p}(\bv)\beta_{p,g}(\bv)
+ \epsilon_{i,p}(\bv),
\end{align*}
where $\epsilon_{i,p}(\bv) \sim N(0,\tau_p^{*2})$, $\tau_p^{*2}=0.01$, and $\rho=0.3$ controls the strength of the interaction.

\subsection{Replication Details}

Each scenario is repeated over 10 independent replicates. For each replicate, training and test datasets are generated independently using the same data-generating mechanism but different random seeds.

\subsection{Coefficient Recovery Evaluation}\label{sup_coef_recovery}

In addition to predictive performance, we assess the ability of AMO-GP to recover the underlying spatially varying coefficients. Since comparable coefficient estimates are not available for the competing methods, this evaluation is restricted to AMO-GP.

For each subnetwork $p = 1, \ldots, P$, we quantify estimation accuracy for the functional coefficients $\beta_{p,g}(\bv)$ and $\beta_{p,w}(\bv)$ using the averaged relative absolute error,
\begin{equation}
\label{eq:rel-abs-error-supp}
\mathcal{E}_{\beta} 
= \frac{1}{2P}\sum_{p=1}^{P}\left\{
\frac{1}{|\mathcal{R}_p|}\sum_{\bv \in \mathcal{R}_p}
\frac{\left|\widehat{\beta}_{p,g}(\bv) - \beta_{p,g}(\bv)\right|}
{\left|\beta_{p,g}(\bv)\right|}
+
\frac{1}{|\mathcal{R}_p|}\sum_{\bv \in \mathcal{R}_p}
\frac{\left|\widehat{\beta}_{p,w}(\bv) - \beta_{p,w}(\bv)\right|}
{\left|\beta_{p,w}(\bv)\right|}
\right\}.
\end{equation}

This metric provides a normalized measure of estimation error across spatial locations and subnetworks, allowing comparison across scenarios with varying signal strength.

\begin{figure}[H]
    \centering
    \includegraphics[width=0.8\linewidth]{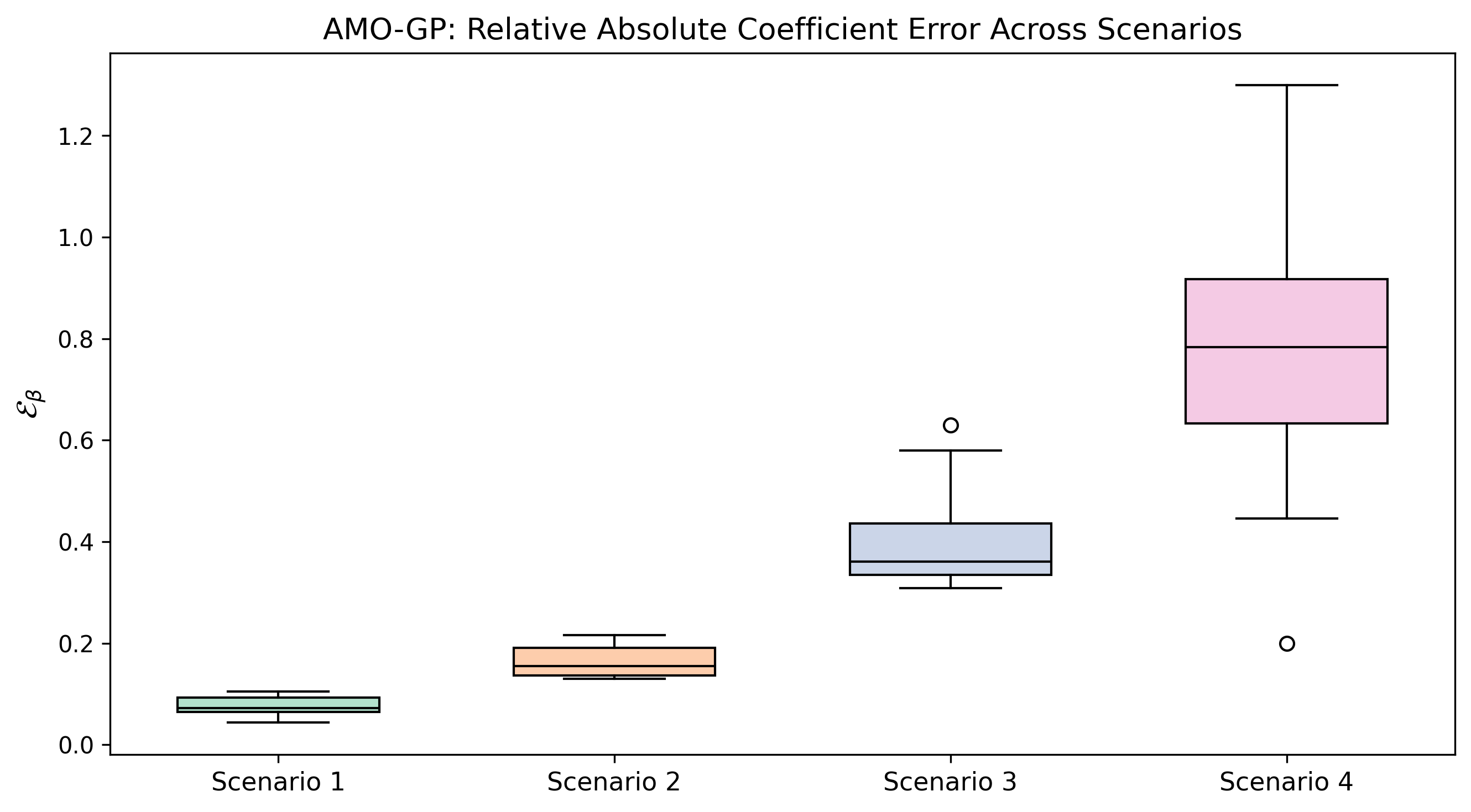}
    \caption{Boxplots of the average relative absolute coefficient estimation error for AMO-GP across 10 replicates under four simulation scenarios. Errors increase and become more variable under misspecification (Scenarios 3 and 4), with visible outliers reflecting heavy-tailed noise and interaction effects.}
    \label{fig:average-relative-coefficient-error}
\end{figure}

Figure~\ref{fig:average-relative-coefficient-error} summarizes the distribution of $\mathcal{E}_{\beta}$ across simulation replicates for all four scenarios. As expected, estimation error is lowest under Scenario 1, where the data-generating mechanism is correctly specified. Scenario 2 exhibits moderate increases in error due to nonlinear misspecification in the network effect. Scenario 3 shows increased variability driven by heavy-tailed noise, while Scenario 4 yields the largest errors due to the presence of non-additive interaction effects not captured by the model.

Overall, these results indicate that AMO-GP is able to accurately recover spatially varying coefficients under correct specification and mild deviations, but performance degrades as model misspecification becomes more severe.

\newpage
\section{Additional Results for ABCD Data Analysis}\label{additional-data-study}
\subsection{Traceplots indicating Stage 2 convergence of AMO-GP}
\label{sup-subsec-ABCDtraceplots}
\noindent For a representative subsample of size 40 of the ABCD data Figure~\ref{traceplot} shows the traceplots for all parameters corresponding to each subnetwork. Each plot shows satisfactory convergence of the MCMC chain.

\begingroup
\centering

\includegraphics[width=.48\textwidth]{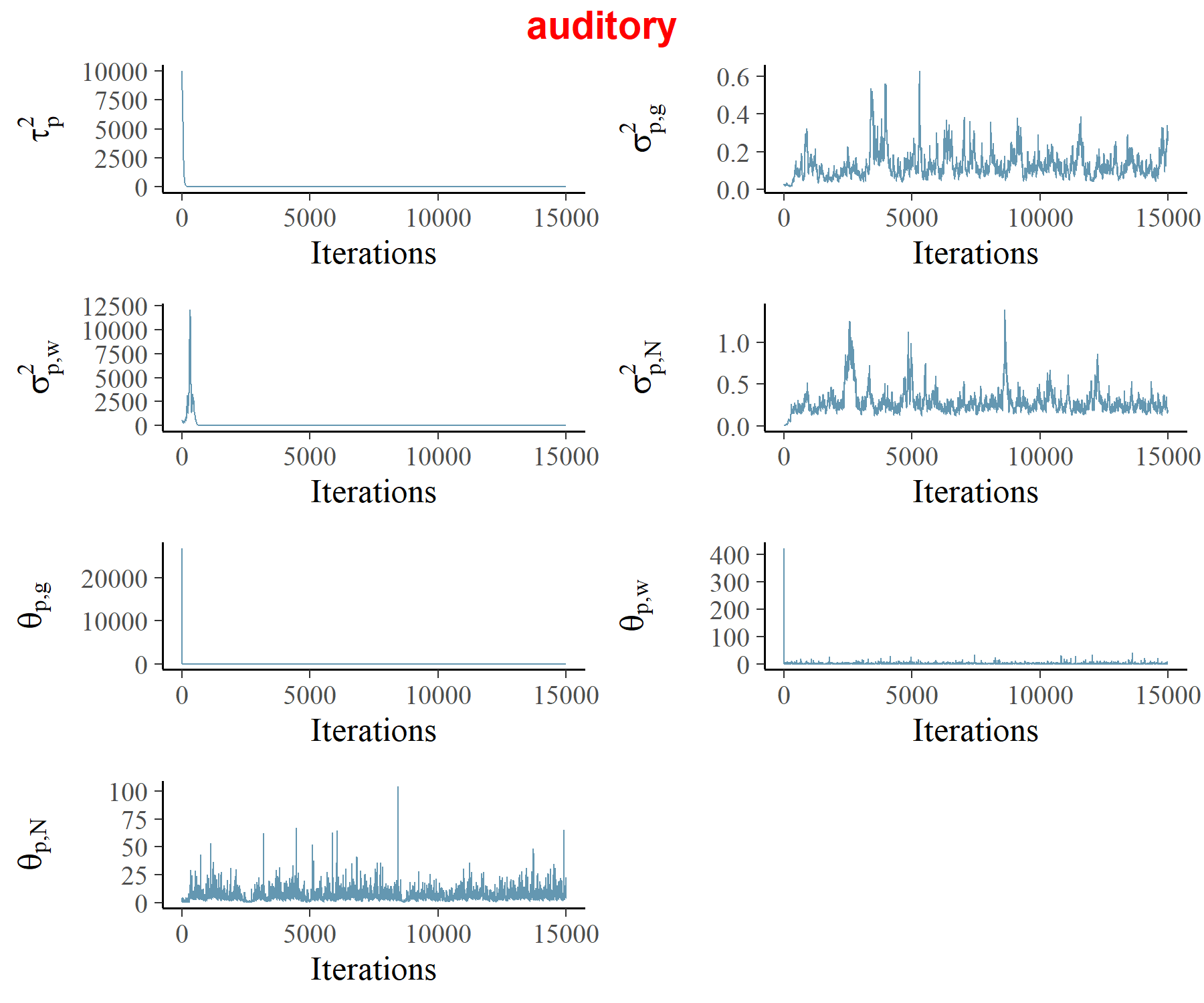}\hfill
\includegraphics[width=.48\textwidth]{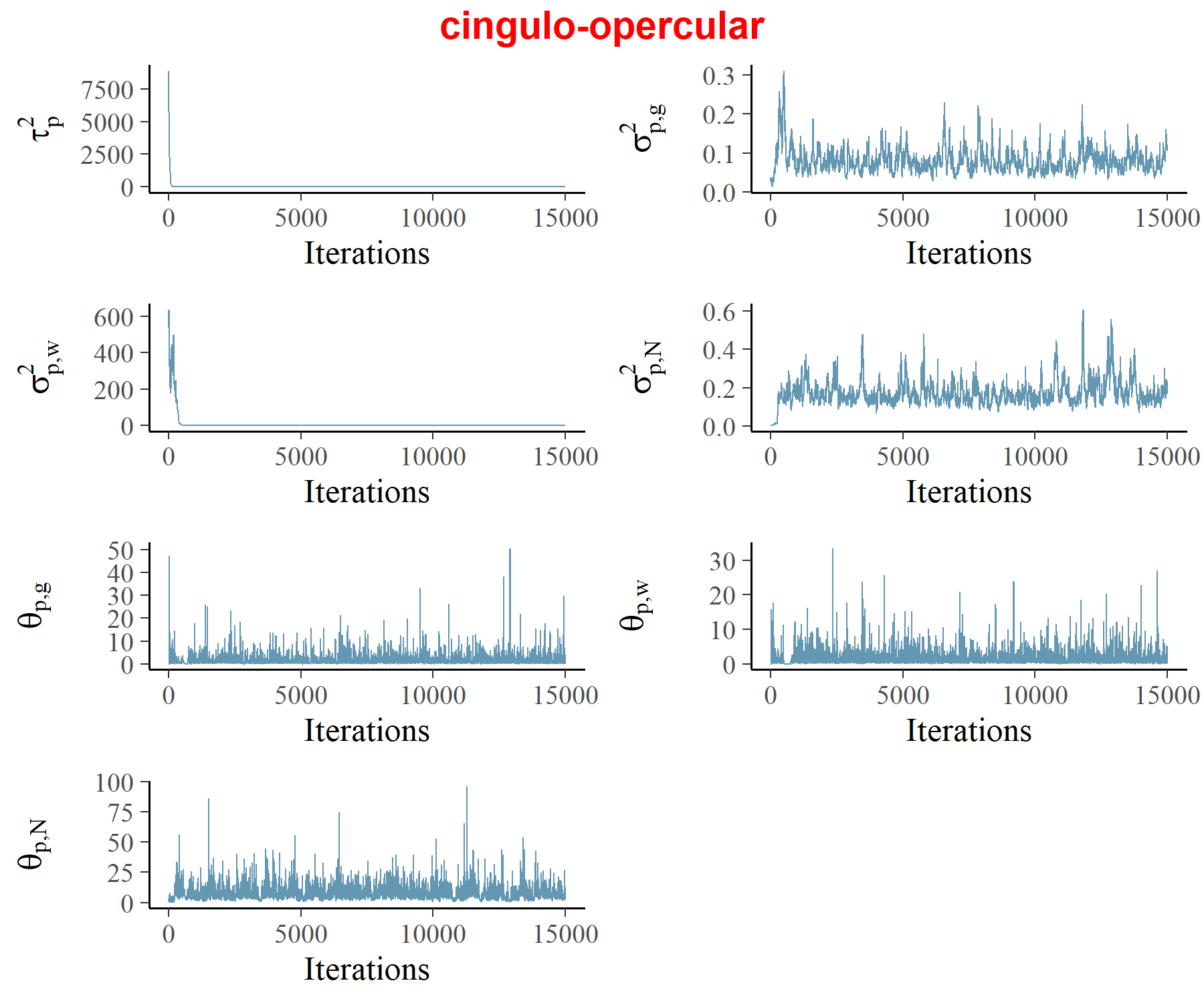}

\vspace{\smallskipamount}

\includegraphics[width=.48\textwidth]{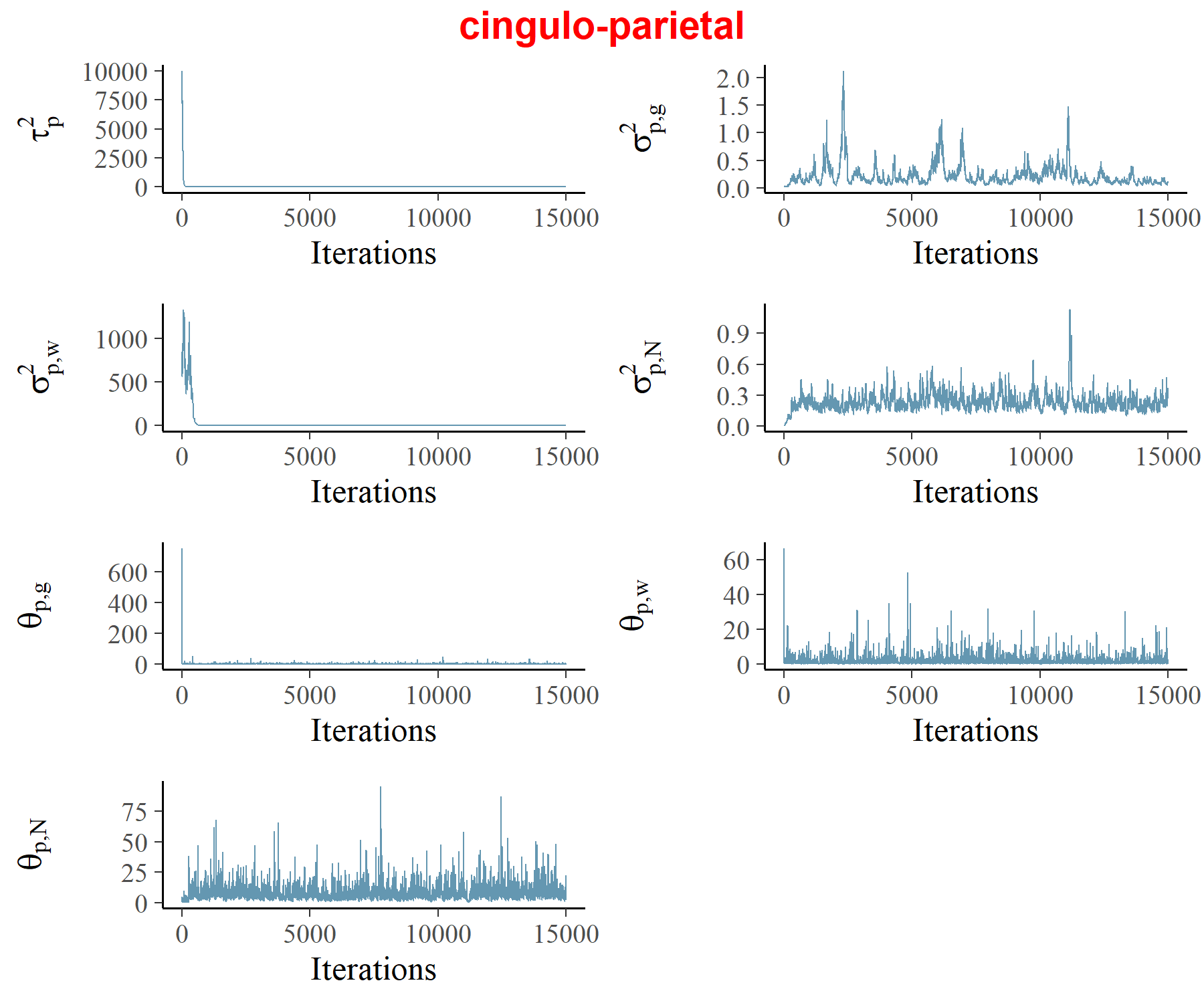}\hfill
\includegraphics[width=.48\textwidth]{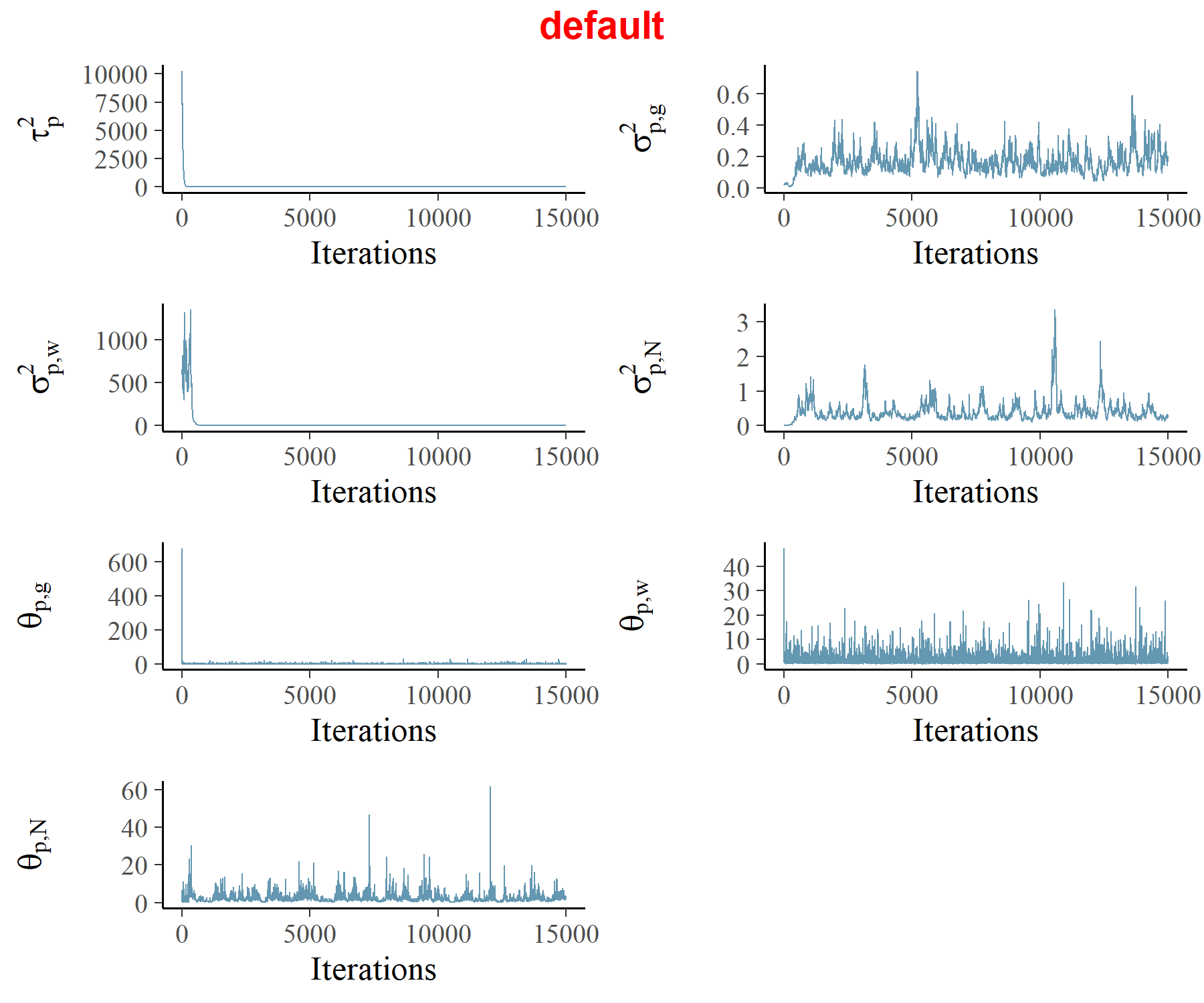}

\vspace{\smallskipamount}

\includegraphics[width=.48\textwidth]{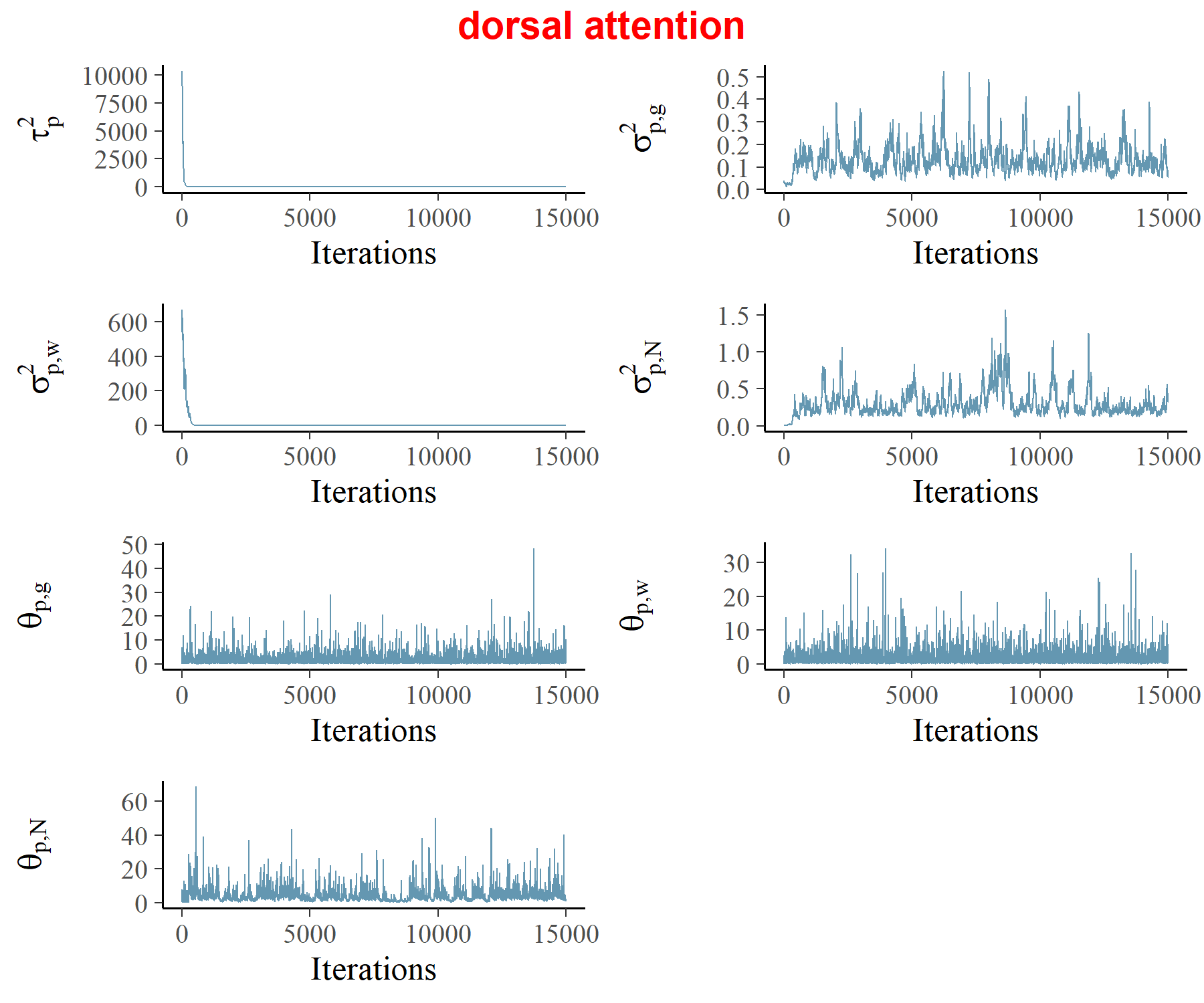}\hfill
\includegraphics[width=.48\textwidth]{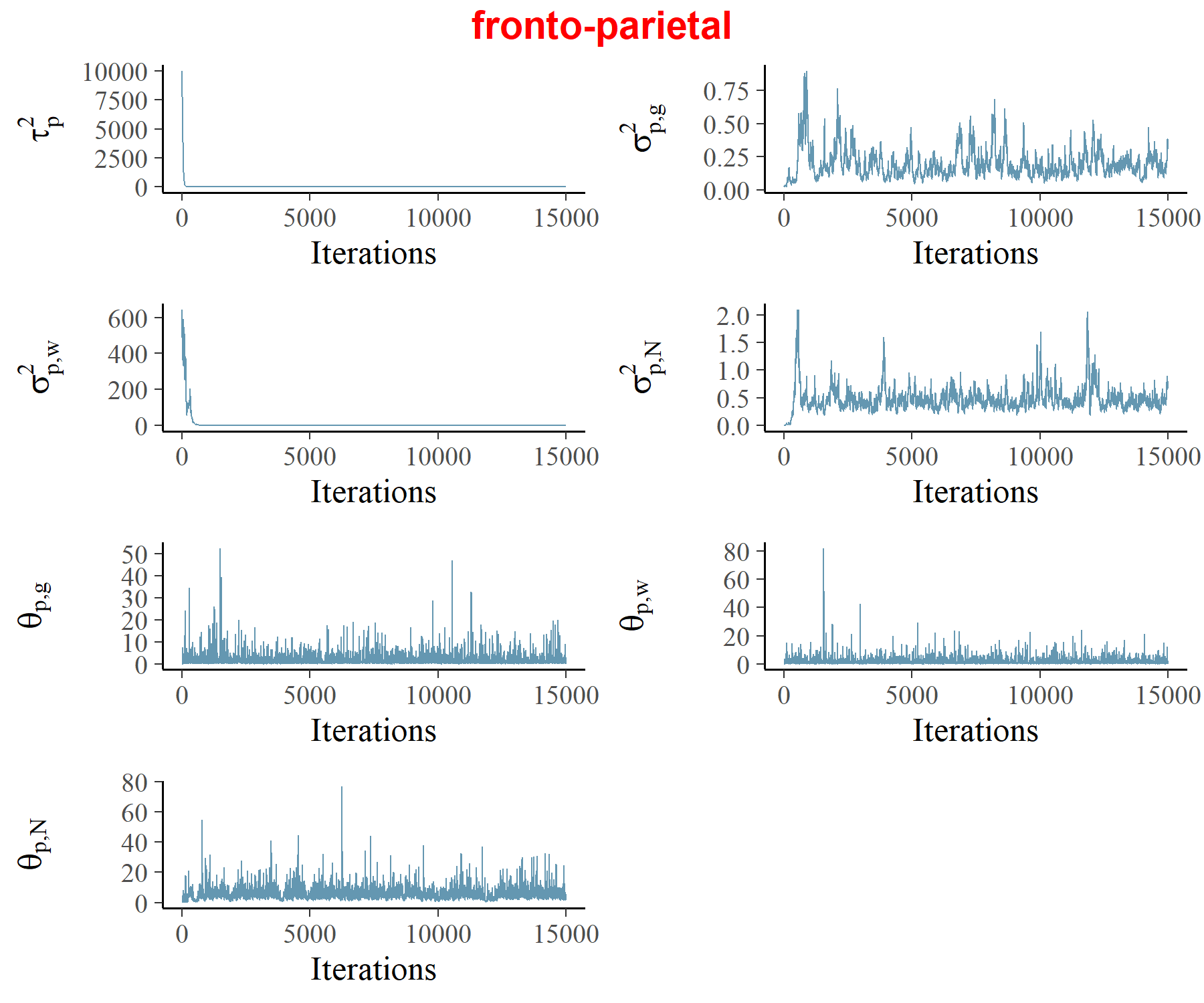}

\vspace{\smallskipamount}

\includegraphics[width=.48\textwidth]{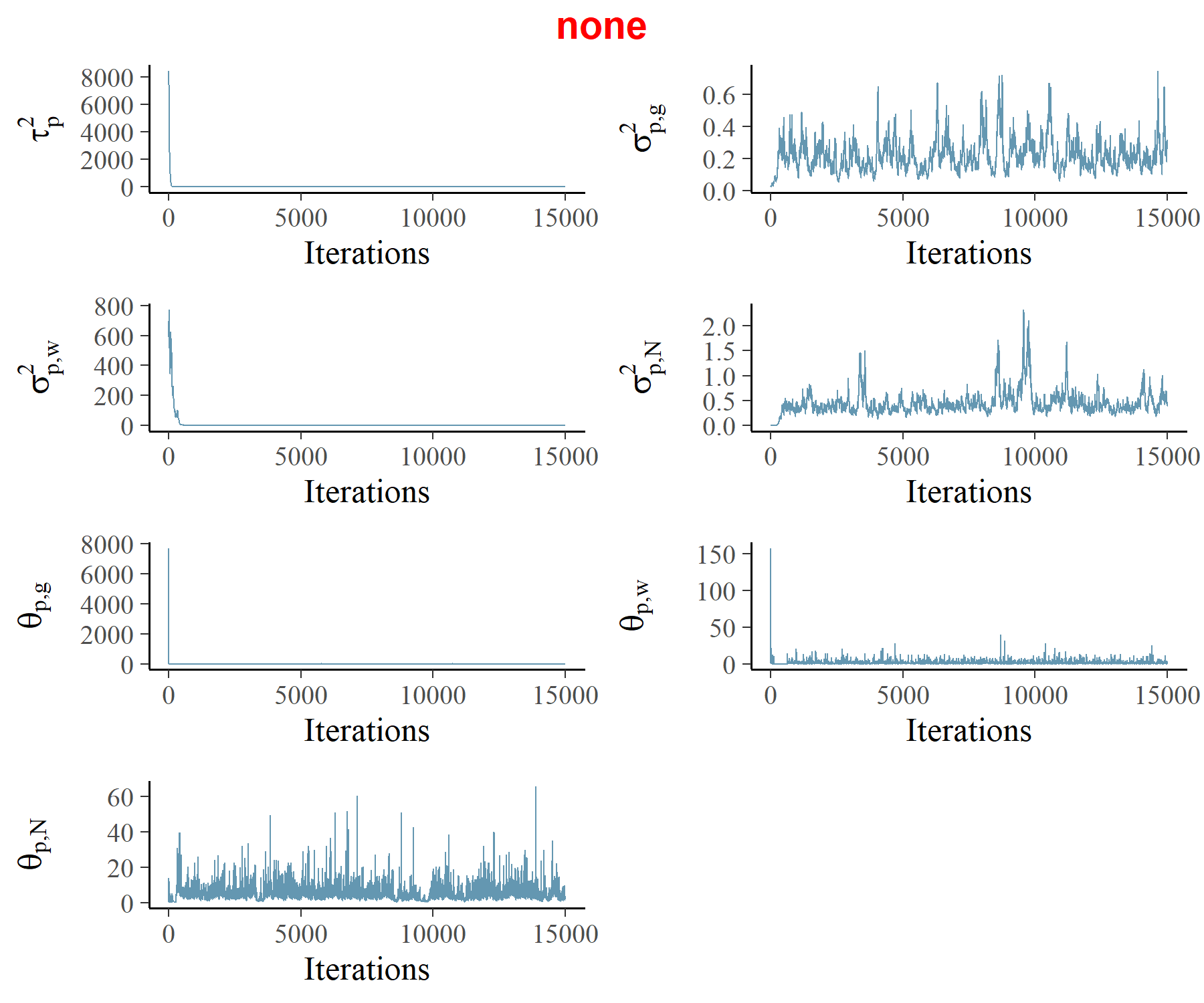}\hfill
\includegraphics[width=.48\textwidth]{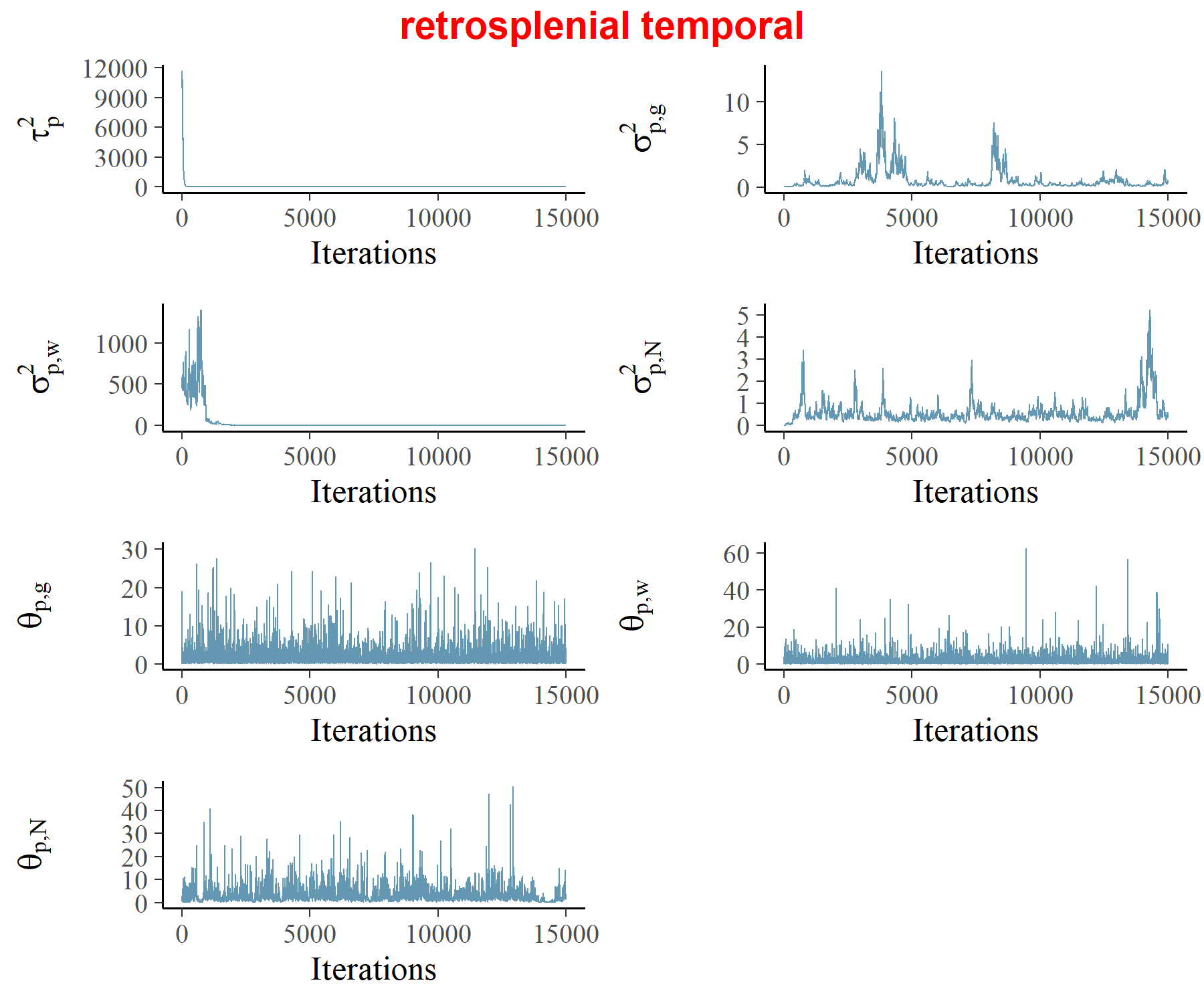}

\vspace{\smallskipamount}

\includegraphics[width=.48\textwidth]{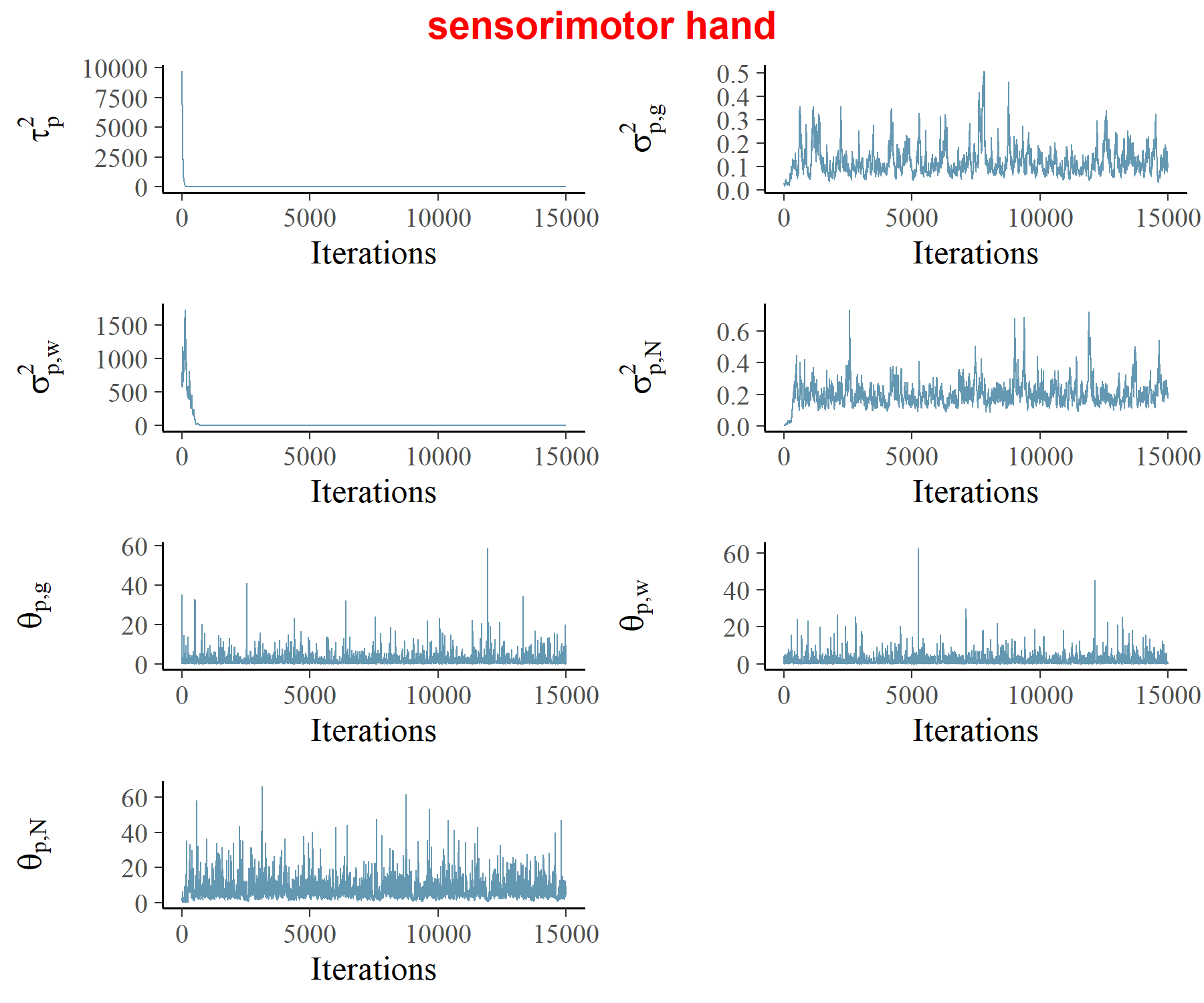}\hfill
\includegraphics[width=.48\textwidth]{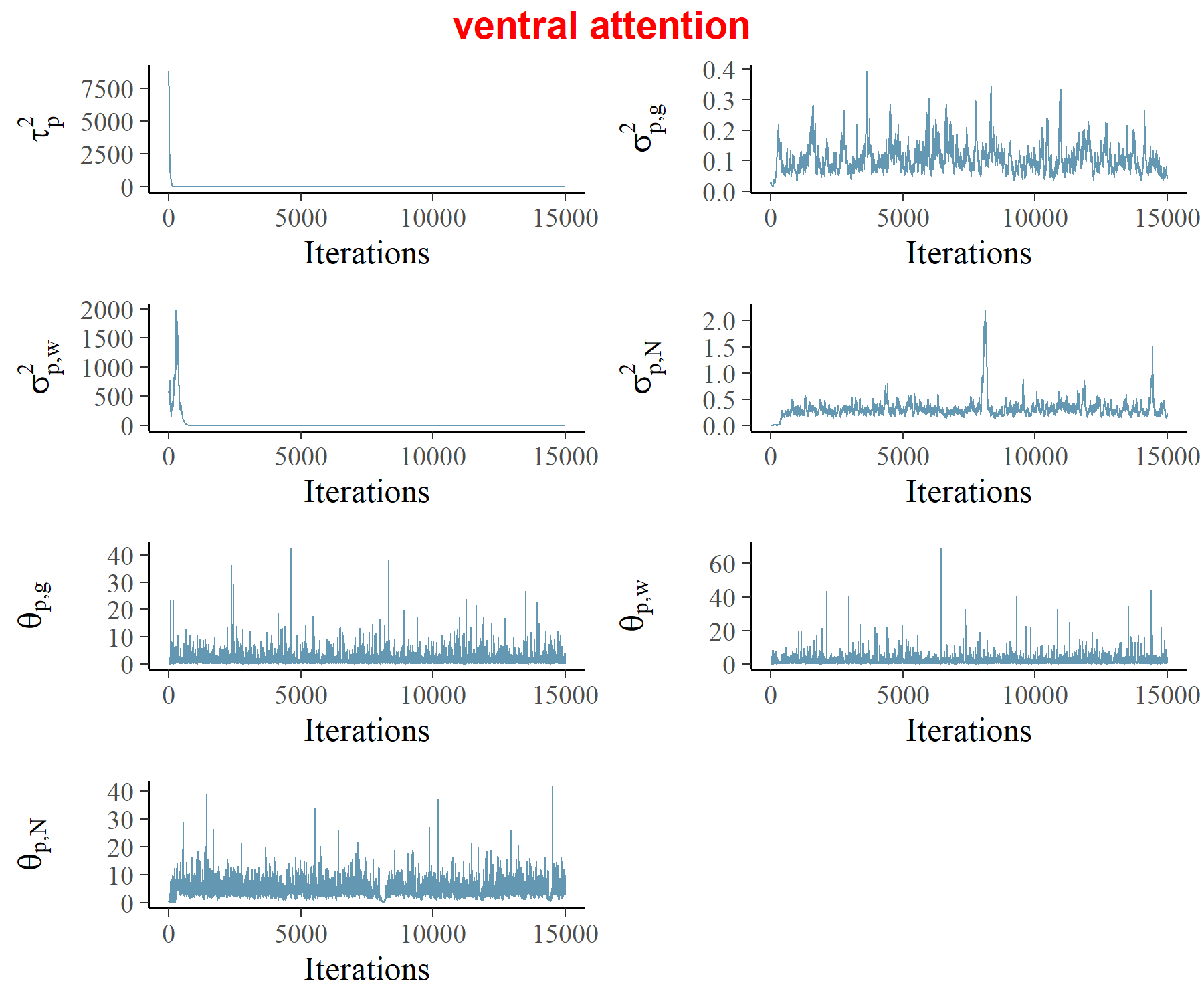}

\vspace{\smallskipamount}

\includegraphics[width=.48\textwidth]{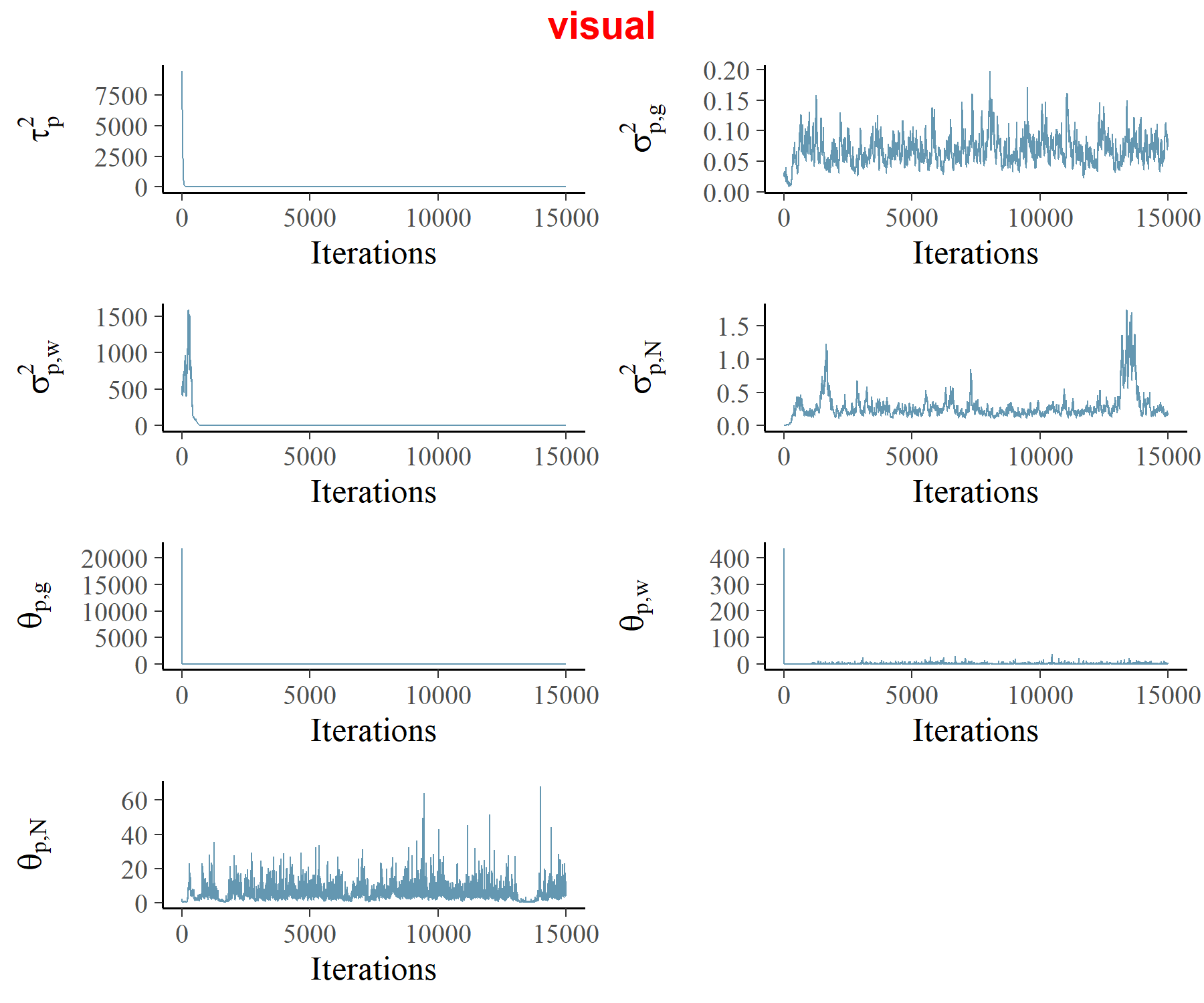}

\captionof{figure}{Trace plots of all parameters for the proposed AMO-GP model across subnetworks over 15,000 iterations for a single subsample ($n=40$). Each subnetwork contains 7 parameters, and results are shown for all 11 subnetworks.}
\label{traceplot}

\endgroup
\vspace{30mm}

\newpage
\subsection{Analysis of in-sample \texorpdfstring{$R^2$}{R-squared} across multiple subsamples}
\label{sup-subsec-insampleR2}

To assess the robustness of AMO-GP, we performed an additional analysis by generating 10 random subsamples of size equal to the original training set (i.e., $40$) and fitting AMO-GP to each. For each fitted subsample, we computed the in-sample $R^2$ across all subnetworks, with results summarized via boxplots in Figure~\ref{r2_subsample_fig}. These results provide insight into the stability of model performance across different realizations of the training data. Across subnetworks, mean $R^2$ values range from approximately $0.32$ (none) to $0.87$ (retrosplenial temporal), with moderate variability as indicated by standard deviations between $0.04$ and $0.11$. Subnetworks such as retrosplenial temporal ($0.87 \pm 0.06$), cingulo-parietal ($0.78 \pm 0.04$), and sensorimotor hand ($0.69 \pm 0.11$) exhibit consistently high $R^2$ values with relatively low dispersion, indicating stable and strong in-sample predictive performance across subsamples. In contrast, subnetworks such as none ($0.32 \pm 0.08$) and fronto-parietal ($0.35 \pm 0.08$) show lower mean $R^2$ with comparable variability, suggesting weaker and less stable fits. The remaining subnetworks, including auditory ($0.66 \pm 0.11$), cingulo-opercular ($0.59 \pm 0.09$), default ($0.54 \pm 0.07$), dorsal attention ($0.53 \pm 0.09$), ventral attention ($0.49 \pm 0.07$), and visual ($0.58 \pm 0.06$), demonstrate moderate predictive performance with intermediate variability. Overall, these results indicate that AMO-GP achieves robust and stable performance across subsamples, with clear heterogeneity in predictive strength across functional subnetworks.

\begin{figure}[H]
\centering
\includegraphics[width=0.8\textwidth]{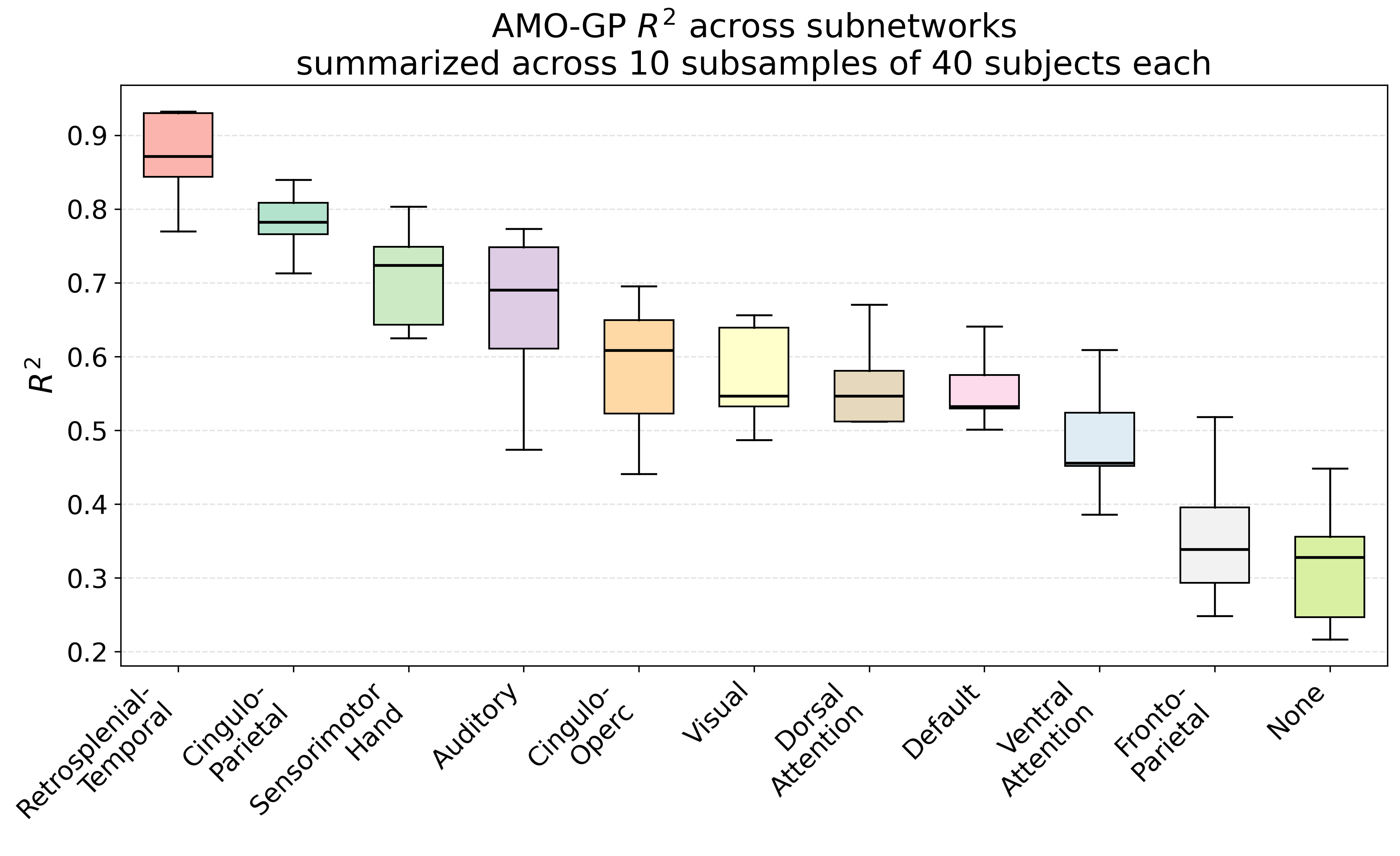}
\caption{Boxplots of in-sample $R^2$ values across functional subnetworks for AMO-GP, obtained from 10 independent subsamples of size 40. Subnetworks are ordered by median $R^2$.}
\label{r2_subsample_fig}
\end{figure}

\newpage
\subsection{Mean squared error (in-sample and predictive) for each subnetwork} \label{sup-subsec-ABCD-MSE}
\noindent Table~\ref{networkmse} shows in-sample mean squared error (MSE) and out-of-sample mean squared prediction error (MSPE) for all competing models for each subnetwork in the ABCD data analysis.
\vspace{0.25in}
\begin{table}[H]
\centering
\small
\caption{Functional subnetwork-level MSE and MSPE across models.}
\label{networkmse}
\begin{tabular}{llccccccc}
\toprule
Network & Metric & AMO-GP & CT+Net & SD+Net & Net & BART & NN & BIRD-GP \\
\midrule

auditory 
& MSE  & 0.10 & 0.13 & 0.13 & 0.13 & 0.15 & 0.36 & 0.25\\
& MSPE & 0.34 & 0.38 & 0.38 & 0.38 & 0.78 & 4.67 & 1.34\\
\midrule

cingulo-opercular 
& MSE  & 0.15 & 0.17 & 0.17 & 0.18 & 0.18 & 0.32 & 0.14\\
& MSPE & 0.35 & 0.37 & 0.37 & 0.37 & 0.66 & 5.40 & 0.97\\
\midrule

cingulo-parietal 
& MSE  & 0.06 & 0.06 & 0.06 & 0.07 & 0.10 & 0.30 & 0.58\\
& MSPE & 0.45 & 0.45 & 0.46 & 0.45 & 0.65 & 5.51 & 3.90\\
\midrule

default 
& MSE  & 0.14 & 0.18 & 0.19 & 0.22 & 0.23 & 0.39 & 0.20\\
& MSPE & 0.72 & 0.77 & 0.74 & 0.79 & 1.04 & 5.97 & 1.58\\
\midrule

dorsal attention 
& MSE  & 0.13 & 0.14 & 0.14 & 0.15 & 0.16 & 0.34 & 0.28\\
& MSPE & 0.64 & 0.64 & 0.64 & 0.62 & 0.88 & 5.99 & 1.04\\
\midrule

fronto-parietal 
& MSE  & 0.66 & 0.67 & 0.68 & 0.71 & 0.67 & 1.08 & 0.41\\
& MSPE & 4.19 & 4.16 & 4.18 & 4.15 & 4.60 & 10.28 & 15.29\\
\midrule

none 
& MSE  & 1.37 & 1.39 & 1.43 & 1.46 & 1.09 & 1.81 & 0.30\\
& MSPE & 4.46 & 4.47 & 4.47 & 4.48 & 4.85 & 9.19 & 10.95\\
\midrule

retrosplenial temporal 
& MSE  & 0.11 & 0.13 & 0.14 & 0.14 & 0.30 & 0.58 & 1.14\\
& MSPE & 1.25 & 1.29 & 1.27 & 1.28 & 1.49 & 7.23 & 3.08\\
\midrule

sensorimotor hand 
& MSE  & 0.05 & 0.06 & 0.06 & 0.06 & 0.11 & 0.25 & 0.26\\
& MSPE & 0.18 & 0.20 & 0.20 & 0.21 & 0.49 & 4.31 & 0.90\\
\midrule

ventral attention 
& MSE  & 0.17 & 0.18 & 0.17 & 0.19 & 0.21 & 0.46 & 0.35\\
& MSPE & 0.93 & 0.93 & 0.94 & 0.93 & 1.23 & 5.73 & 1.76\\
\midrule

visual 
& MSE  & 0.20 & 0.20 & 0.20 & 0.22 & 0.22 & 0.39 & 0.14\\
& MSPE & 0.79 & 0.79 & 0.80 & 0.81 & 1.27 & 6.15 & 2.44\\

\bottomrule
\end{tabular}
\end{table}

\setstretch{1}

\putbib[bibliography_paper]

\end{bibunit}

\end{document}